\documentclass[submission]{eptcs}
\providecommand{\event}{} 

\usepackage{etoolbox}
\newtoggle{colorON}
\toggletrue{colorON}

\newtoggle{colorElsON}
\toggletrue{colorElsON}

\newtoggle{hidingON}
\toggletrue{hidingON}

\newtoggle{headOutlineON}
\toggletrue{headOutlineON}

\newtoggle{mNotsON}
\toggletrue{mNotsON}

\newtoggle{appendixON}
\togglefalse{appendixON}

\usepackage{etex}
\usepackage[usenames]{color}
\usepackage{verbatim}
\usepackage{fancyhdr}
\usepackage{fancybox}
\usepackage{stmaryrd}
\usepackage[reqno]{amsmath}
\usepackage{amsthm}
\usepackage{amssymb}
\usepackage{amsfonts}
\usepackage{bm}
\usepackage{amsbsy}
\usepackage[all]{xypic}
\usepackage {indentfirst}
\usepackage{graphicx}
\usepackage[hang]{subfigure}
\usepackage{cite}
\usepackage{proof}
\usepackage{bbding}
\usepackage{pgf}
\usepackage{tikz}
\usetikzlibrary{positioning,arrows,automata}
\usepackage{cancel}
\usepackage{extarrows}
\usepackage{etoolbox} 
\usepackage{hyperref} 
\hypersetup{%
 dvipdfmx,
 unicode={true},
 pdfstartview={FitH},
 bookmarksnumbered={true},
 bookmarksopen={true},
 pdfborder={0 0 0},
 colorlinks = true, 
}
\usepackage{algorithm}
\usepackage{algorithmicx}
\usepackage{algpseudocode}

\newtheorem{theorem}{Theorem}
\newtheorem{proposition}[theorem]{Proposition}
\newtheorem{lemma}[theorem]{Lemma}
\newtheorem{corollary}[theorem]{Corollary}

\theoremstyle{definition}
\newtheorem{definition}[theorem]{Definition}

\theoremstyle{remark}

\newcommand{\para}{\,|\,}
\newcommand{\fn}[1]{\mbox{\rm fn($#1$)}} 
\newcommand{\n}[1]{\mbox{\rm n($#1$)}} 
\newcommand{\nc}[1]{\mbox{\rm nc($#1$)}} 
\newcommand{\fnv}[1]{\mbox{\rm fnv($#1$)}} 
\newcommand{\bnv}[1]{\mbox{\rm bnv($#1$)}} 
\newcommand{\nv}[1]{\mbox{\rm nv($#1$)}} 
\newcommand{\fpv}[1]{\mbox{\rm fpv($#1$)}} 
\newcommand{\bpv}[1]{\mbox{\rm bpv($#1$)}} 
\newcommand{\pv}[1]{\mbox{\rm pv($#1$)}} 

\newcommand{\fosub}[2]{\{#1/#2\} }
\newcommand{\hosub}[2]{\{#1/#2\} }
\newcommand{\ve}[1]{\widetilde{#1}}

\newcommand{\st}[1]{\,{\xrightarrow{#1}}\, }

\newcommand{\rc}[1]{{\color{red} #1}}
\newcommand{\bc}[1]{{\color{blue} #1}}
\newcommand{\R}{\ensuremath{\mathcal{R}} }
\newcommand{\sepp}{\vspace*{0.4cm}}

\newcommand{\nsepv}[1]{\vspace{0mm}}

\newcommand{\TODO}{$\blacktriangleright\blacktriangleright\blacktriangleright\blacktriangleright\blacktriangleright\blacktriangleright\blacktriangleright\blacktriangleright\blacktriangleright\blacktriangleright\blacktriangleright\blacktriangleright\blacktriangleright\blacktriangleright\blacktriangleright\blacktriangleright\blacktriangleright\blacktriangleright\blacktriangleright\blacktriangleright\blacktriangleright\blacktriangleright$ \textsc{\large TODO!!!!!!....} }
\newcommand{\TODOM}[1]{\xxa{$\blacktriangleright\blacktriangleright\blacktriangleright\blacktriangleright\blacktriangleright\blacktriangleright\blacktriangleright\blacktriangleright\blacktriangleright\blacktriangleright\blacktriangleright\blacktriangleright\blacktriangleright\blacktriangleright\blacktriangleright\blacktriangleright\blacktriangleright\blacktriangleright\blacktriangleright\blacktriangleright\blacktriangleright\blacktriangleright$ \textsc{\large TODO!!!!!!....} #1 } }

\newcommand{\SE}{\equiv }
\newcommand{\DEF}{\stackrel{\textrm{def}}{=}} 
\newcommand{\lrangle}[1]{\langle #1 \rangle} 
\makeatletter
\@ifundefined{newslide}{%

}{%
}
\makeatletter
\@ifundefined{endinput}{%

}{%
}

\newcommand{\xxa}[1]{{\textcolor[RGB]{185,2,1}{#1}}}

\newcommand{\xxy}[1]{{\textcolor[RGB]{0,205,175}{#1}}}
\newcommand{\xxyrmcolor}[1]{#1} 

\newcommand{\zwbrmcolor}[1]{#1}

\newcommand{\ntsrm}[1]{} 

\newcommand{\ntsoorm}[1]{} 

\newcommand{\tdup}[1]{} 
\newcommand{\erase}[1]{} 
\newcommand{\caseHintRM}[1]{} 
\newcommand{\annotate}[1]{} 

\newcommand{\HOmp}{\ensuremath{\Pi^{\mbox{\tiny mp}}}} 
\newcommand{\BHOParam}{\ensuremath{\Pi^{-}_{D,d}} } 
\newcommand{\SCongru}{\ensuremath{\equiv} } 

\newcommand{\SHOIOB}{\ensuremath{\sim_{\mbox{\tiny hoio}}^{\mbox{\tiny $\circ$}}} } 
\newcommand{\SHOB}{\ensuremath{\sim_{\mbox{\tiny ho}}} } 
\newcommand{\SCTXB}{\ensuremath{\sim_{\mbox{\tiny ctx}}} } 
\newcommand{\SNRB}{\ensuremath{\sim_{\mbox{\tiny nr}}} } 
\newcommand{\OSNRB}{\ensuremath{\sim_{\mbox{\tiny nr}}^{\mbox{\tiny $\circ$}}} } 

\newcommand{\mtrigger}{\ensuremath{\overline{m}} } 
\newcommand{\mtriggerD}{\ensuremath{\lrangle{Z}\overline{m}Z} } 
\newcommand{\mtriggerDd}{\ensuremath{\lrangle{z}\overline{m}[\lrangle{Z}(Z\lrangle{z})] }} 

\newcommand{\mtriggerName}{\ensuremath{{\rm Tr}_m} } 
\newcommand{\mtriggerDName}{\ensuremath{{\rm Tr}_m^{D}} } 
\newcommand{\mtriggerDdName}{\ensuremath{{\rm Tr}_m^{D,d}} } 

\newcommand{\genTrigger}[2][m]{\ensuremath{{\rm Tr}_{#1}^{#2}}} 

\newcommand{\pdepth}[1]{\ensuremath{{\rm depth}(#1)} } 

\newcommand{\normalform}[1]{\ensuremath{{\rm nf}(#1)} }
\newcommand{\Nil}{0}

\newcommand{\rewriteDIS}{\rightsquigarrow}
\newcommand{\SB}{\ensuremath{\sim} } 
\newcommand{\eat}[1]{} 
\newcommand{\warn}[1]{} 

\newcommand{\db}{\textsf{db}}
\newcommand{\tree}{\textsf{Tree}}
\newcommand{\tnAbs}{\textsf{abs}}
\newcommand{\tnApp}{\textsf{app}}
\newcommand{\labelI}[1]{{#1}}
\newcommand{\labelO}[1]{{#1}^{O}}

\usepackage{multicol}
\RequirePackage[normalem]{ulem}
\RequirePackage{color}\definecolor{RED}{rgb}{1,0,0}\definecolor{BLUE}{rgb}{0,0,1}

\title{On Decidability of the Bisimilarity on Higher-order Processes with Parameterization\thanks{This work is supported by ANR 12IS02001 PACE, NSF of China (61872142, 62072299, 61772336, 61572318, 61261130589), Shanghai Sailing Program (21YF1417000) and the Open Project of Shanghai Key Laboratory of Trustworthy Computing.}
}
\author{Xian Xu
\institute{East China University of Science and Technology}
\email{xuxian@ecust.edu.cn}
\and
Wenbo Zhang 
\institute{Shanghai Ocean University \\
Shanghai Key Laboratory of Trustworthy Computing}
\email{wbzhang@shou.edu.cn}
}
\def\titlerunning{On Decidability of the Bisimilarity on Higher-order Processes with Parameterization}
\def\authorrunning{Xian Xu \& Wenbo Zhan}

\begin{document}
\maketitle

\begin{abstract}
\noindent\emph{\textbf{Abstract}}~ 
Higher-order processes with parameterization are capable of abstraction and application (migrated from the lambda-calculus), and thus are computationally more expressive. For the minimal higher-order concurrency, it is well-known that the strong bisimilarity (i.e., the strong bisimulation equality) is decidable in absence of parameterization. By contrast, whether the strong bisimilarity is still decidable for parameterized higher-order processes remains unclear. In this paper, we focus on this issue. There are basically two kinds of parameterization: one on names and the other on processes. 
We show that the strong bisimilarity is indeed decidable for higher-order processes equipped with both kinds of parameterization. Then we demonstrate how to adapt the decision approach to build an  axiom system for the strong bisimilarity. On top of these results, we provide an algorithm for the bisimilarity checking.

\vspace*{.1cm}

\vspace*{.1cm}
\end{abstract}

%
\section{Introduction}\label{s:introduction}






Bisimulation is a most important concept for comparing the behaviour of computing systems, particularly concurrent systems.  An accompanying vital question is to check whether two given systems are equal in terms of bisimulation, hence the bisimilarity checking. 
Bisimilarity checking is an important 
topic in concurrency theory and formal verification. Basically there are two directions for this topic. One is to adopt an abstract manner, using process rewrite systems \cite{KJ06}. An advantage of this direction is that some core techniques can be extracted and potentially adapted to various models. The other is to work directly on concrete models \cite{Per10}. An edge of this direction is that some well-defined operators can be harnessed thoroughly to guide the checking. We focus on the second direction in this work. 

\tdup{A lot of work has emerged, in the way of both process rewrite system and process calculi \cite{Jan95, Sch01, KJ06, CHM94, KH12}.} 
\tdup{\scriptsize [a few references here; survey in PRS; and process models, first see the references in \cite{LPSS10a})][More???]}

The bisimilarity checking, including checking bisimulation equalities, simulations, and preorders, has been attracting tremendous attention in the past few decades \cite{Jan95, Sch01, KJ06, CHM94, KH12}.
In contrast to the fruitful work of bisimilarity checking on first-order 
 models,
checking bisimulation equalities for higher-order processes has been more challenging. Much fewer results have been known in higher-order process models. Indeed, a major reason is that higher-order processes communicate in the fashion of process-passing (i.e., program-passing), and have the innate capability of encoding recursion. Besides, the standard bisimulation for higher-order processes, i.e., the context bisimulation, is strikingly different from those for first-order processes. It requires the matching of two output processes to be compared in arbitrary contexts. To this point, simplifying the context bisimilarity has also been a significant topic \cite{San92, San96}. 

To this day, the best known result of bisimilarity checking for higher-order processes is reported in \cite{LPSS10a}, to our knowledge. In that work, Lanese et al. show that the strong bisimilarity checking of HOcore processes is decidable. As a matter of fact, they show that all known strong bisimilarities in HOcore are decidable. HOcore is a minimal higher-order process model that only has the input, (asynchronous) output, and the concurrency operator (i.e., the parallel composition). HOcore is also proven to be Turing complete, and this result is somehow refined toward a more implementable interpretation \cite{BBLPPS17}, in the manner of encoding lambda-calculus in HOcore through abstract machines. 

That HOcore is Turing complete renders its process-termination problem undecidable. This fact adds to the contrast that the strong bisimilarity is decidable, which in turn implies the decidability of the barbed congruence. 
Technically, the decidability is achieved by showing all the strong bisimilarities to be coincident with a very special strong bisimilarity, called IO-bisimilarity, which is decidable by its definition in the first place. On the basis of this decidability outcome, a complete axiom system is also established, as well as an algorithm with acceptable complexity. It is then possible to implement the algorithm for bisimilarity checking HOcore processes in software systems \cite{ALS21}.
Intuitively, the essential element making the strong bisimilarity decidable is that HOcore does not have the restriction operator, and thus the capability of expressing recursion is weakened. It is also shown that if restriction is recovered, i.e., if at least four static (i.e., top concurrency level) restrictions are included in HOcore, then the strong bisimilarity immediately becomes undecidable. The undecidability is proven through a reduction from the PCP; similar reductions are also used in other settings, e.g., the Ambient calculus \cite{CT01}. 
Building upon \cite{LPSS10a,GPZ09} further studies the possibility of making the termination decidable, in the setting of a fragment of HOcore where nested higher-order outputs are disabled. Specifically, it is shown that in such a setting the termination of processes becomes decidable (though convergence is still undecidable), due to the reason that the Minsky machines are no longer expressible. Technically, such decidability is achieved using the well-structured transition systems employed in \cite{BGZ09}. Following \cite{GPZ09,Per10} shows that termination turns back to be undecidable if such a fragment of HOcore is enriched with a passivation operator \cite{SS05}, because Turing completeness is retained with the help of passivation.

In \cite{BGHH09}, Bundgaard et al. study the decidable fragments of Homer \cite{BGH04}, a higher order process model with the capacity of expressing locations. They show that two subcalculi of Homer have decidable barbed bisimilarity, in both the strong and weak forms. 
Intuitively, Homer supports certain kind of pattern matching of name sequences that model the locations of resources, and this plays a central role in enhancing the expressiveness. For this reason, Homer can encode first-order processes and is computationally complete, leaving little hope for the decidability of bisimilarities. Therefore, to obtain decidability, some constraints have to be devised. 
Technically, such constraints are imposed through a finite control property. That is, some finite reachability criterion is excerpted on the semantics of Homer processes. Such a criterion is the key reason for the decidability of barbed bisimilarities. 
The approach of \cite{BGHH09} provides a valuable reference for acquiring decidability sub-models from a more powerful full model.

\annotate{
\input{part_of_related_work_1}
}

However, there is still much space one can exploit concerning bisimilarity checking for higher-order processes, as mentioned in \cite{LPSS10a}. HOcore is a minimal model, with somewhat low modelling capacity. It would be interesting to quest for a more expressive model by adding certain constructs, while still maintaining the decidability result. Parameterization has been known to be an effective approach of promoting the expressiveness of higher-order processes, that is, abstraction-passing is strictly more expressive than mere process-passing \cite{LPSS10}. In this work, we focus on the minimal higher-order processes with parameterization, {notation \HOmp}, basically HOcore extended with parameterization. 
This minimal model contains solely the most elementary parts to formalize higher-order concurrency, with extension of the abstraction and application, two operations originating from the lambda-calculus\cite{Bar84}.
We will show that in such a calculus, the strong bisimilarity remains decidable. Similar result is only conjectured in \cite{LPSS10a}.  To this point, we go beyond that conjecture in two respects. Firstly, although our general approach resembles that of \cite{LPSS10a}, the technical route has some key lemmas with essentially different proof structures, due to the complication brought by the parameterization. Secondly, we consider two kinds of parameterization, i.e., both on names and on processes themselves, rather than only one kind. 
\xxyrmcolor{
  Thus we are working on a potentially more expressive model. This is evidenced by the following two facts.
(1) Parameterization, in particular process parameterization, brings strictly more expressiveness to the higher-order process model \cite{LPSS10}.
(2) Moreover, name parameterization is more expressive than process parameterization \cite{YXL17}. Intuitively, this is true because we can somehow encode process parameterization with name parameterization, using an idea akin to that of encoding process-passing into name-passing.}
To the best of our knowledge, there has been little work about the decidability of bisimilarities in such a model. 
The decidability result of this work not only pushes outward the boundary of higher-order processes with decidable bisimilarity, but also digs more into the realm of bisimilarity checking more \zwbrmcolor{challenging} behavioural equalities, such as weak or branching bisimilarity. 
\annotate{On the basis of the decidability of the strong bisimilarity in the minimal higher-order process with parameterization, we build an axiom system and the corresponding bisimilarity checking algorithm, in roughly the same vein of that in \cite{LPSS10a}.}

\annotate{
\input{part_of_related_work_2}
}

\vspace*{.3cm}
\noindent\textbf{Contribution}~~ Now we summarize the main contribution of this paper.
 
\noindent$\bullet$~~ We show that in the minimal higher-order process model with parameterization, the strong bisimilarities, including the standard context bisimilarity together with other well-known bisimulation equalities, are all decidable. 
We borrow and revamp the ideas from \cite{LPSS10a}, i.e., defining a bisimilarity decidable from the very beginning and then showing that the bisimilarities of interest coincide with it. 
The major novel parts are those tackling the parameterization. Due to the presence of the parameterization, we have a completely new design of the key bisimilarities, particularly those defined directly over open processes (i.e., those processes carrying free variables), \xxyrmcolor{as well as the normal bisimulation that needs new forms of triggers for the two kinds of parameterizations}. In turn, the congruence proofs must take these changes into consideration. Moreover, some crucial properties for establishing the coincidence of the bisimilarities have entirely new proof methods, in particular, among others, the preservation of substitution that claims the closure of variable substitutions with respect to the strong bisimilarity (since now a variable can take an abstraction). Indeed, the discussion of the mutual inclusion of various bisimilarities calls for more rigorous and fine-grained investigation in the setting of parameterization. More explanation is given in Sections \ref{s:preliminary}, \ref{s:dec_big_d_small_d}.\\
\noindent$\bullet$~~ With the decidability in place, we design an axiom system and a checking algorithm, in roughly the same vein as those in \cite{LPSS10a}, with the following \zwbrmcolor{differences}. 
(1) For the axiom system, the core part amounts to reducing the deduction of the strong bisimilarity to the extended structural congruence. Previously, such extension  includes a distribution law. Now with parameterization in the game, we have to further extend the structural congruence with the laws for the application operation. 
(2) For the bisimilarity-checking algorithm, the core is to transform a term (possibly with parameterization) into certain normal form with the help of a tree representation of the process, and then the bisimilarity checking can be readily done almost syntactically on the normal form. In presence of parameterization, we extend the tree to accommodate abstractions and applications, as well as the normalization procedure. 
In such an extended procedure, we execute applications as many times any possible, and operate the tree in a bottom-up fashion so as to improve on performance.  
The algorithm has \xxyrmcolor{linear space complexity} and polynomial time complexity slightly better than available ones. More details are given in Sections \ref{s:axiomatization}, \ref{s:dec_algorithm}.\\
\annotate{\rc{[\emph{\Large ...maybe rewrite to be consistent...}]Beyond that, we improve on two more respects. Firstly, we avoid using De Bruijn indices to remove variables, because (arguably) they are not part of the process model. Instead, the variables are  uniformly renamed. That is, regarding the tree representation of a process, we attach a same bound variable name to each level of the tree. This would not lead to conflict. This renaming pre-processing is executed at the beginning of the algorithm. 
Secondly, to make the tree more compact after normalization, we perform cleaning more frequently (in the algorithm for HOcore, this is done only once). That is, every time we transform the tree representation of a process, we do garbage collection by removing those sub-processes equal to $0$. As a result, this operation would render more efficient the comparison of trees, and the checking algorithm on the whole as well.
[...]
} 
}
\iftoggle{appendixON}{%
} {%
\xxyrmcolor{An extended version of this paper with more details is available \cite{XZ21L}.}
}


\noindent\textbf{Organization} The remainder of this paper is organized as follows. 
Section \ref{s:preliminary} gives the definitions of the process model and the strong bisimilarities. 
Section \ref{s:dec_big_d_small_d} presents the decidability of the strong bisimilarities, with detailed proofs. 
Section \ref{s:axiomatization} does the axiomatization and proves its correctness. 
In Section \ref{s:dec_algorithm}, we demonstrate an algorithm for the bisimilarity checking, and analyses its complexity. 
Section \ref{s:conclusion} concludes this paper and points to some future work.





\vspace*{-.4cm}
\section{Preliminary}\label{s:preliminary}  
In this section, we first define \HOmp, the minimal higher-order process model extended with parameterization. Then we introduce the strong bisimilarities to be discussed.

%

\iftoggle{mNotsON}{%
} {%
\TODO
\TODOM{Notice the part on parameterization}
\TODOM{Explain and convention...} 
}


\noindent\textbf{Syntax}~~ Calculus \HOmp\ has the following syntax.  

\annotate{
\[
\begin{array}{lcll}
P, Q &:=& 0 \qquad & \mbox{nil} \\
&\,\Big{|}\,& X & \mbox{process variable} \\
&\,\Big{|}\,& m(X).P & \mbox{input prefix} \\
&\,\Big{|}\,& \overline{m}(Q) & \mbox{output} \\
&\,\Big{|}\,& P\para Q & \mbox{parallel composition} \\ 
&\,\Big{|}\,& \lrangle{X}P & \mbox{process abstraction} \\  
&\,\Big{|}\,& P\lrangle{Q} & \mbox{process  application} \\ 
&\,\Big{|}\,& \lrangle{x}P & \mbox{name abstraction} \\
&\,\Big{|}\,& P\lrangle{n} & \mbox{name application} \\
\end{array}
\]
}

$\hspace*{3cm}
P, Q := 0 
\,\Big{|}\, X
\,\Big{|}\, m(X).P 
\,\Big{|}\, \overline{m}(Q) 
\,\Big{|}\, P\para Q 
\,\Big{|}\, \lrangle{X}P   
\,\Big{|}\, P\lrangle{Q} 
\,\Big{|}\, \lrangle{x}P 
\,\Big{|}\, P\lrangle{n} 
$

\HOmp\ expressions (or terms, processes) are represented by capital letters. For the sake of convenience, we divide names (ranged over by $m,n,u,v...$) into two groups: one for name constants (ranged over by $a,b,c,d,e...$) and the other for name variables (ranged over by $x,y,z...$).
The elements of the calculus have their standard meaning. 
One notices that the output is non-blocking, i.e., asynchronous. Sometimes we write $\overline{m}[Q]$ for output. 
Input $m(X).P$ and process abstraction $\lrangle{X}P$ bind the process variable $X$, and name abstraction $\lrangle{x}P$ binds the name variable $x$. Otherwise, a process or name variable is free. Bound variables can be replaced subject to $\alpha$-conversion, and the resulting term is deemed as the same. 
A term is closed if it does not have free \xxyrmcolor{process} variables. 
Otherwise it is open.
\annotate{We sometimes simply refer to closed terms as processes, and (general) terms as open processes. Terms $P\lrangle{Q}$ and $P\lrangle{n}$ are process application and name application respectively.
}
Operations $\fpv{\cdot}$, $\bpv{\cdot}$, $\pv{\cdot}$, $\fnv{\cdot}$, $\bnv{\cdot}$, $\nv{\cdot}$, $\nc{\cdot}$, $\n{\cdot}$ respectively return the free process variables, bound process variables, process variables, free name variables, bound name variables, name variables, name constants, and names of a set of terms. 
A variable or name is fresh if it does not appear in the terms under examination. 
We use $\ve{\cdot}$ for a tuple, for example, a tuple of terms $\ve{P}$ and a tuple of names $\ve{m}$. 
Process substitution $P\hosub{Q}{X}$ (respectively name substitution $P\hosub{m}{x}$) denotes the replacement of process variable $X$ (respectively name variable $x$) with the process $Q$ (respectively name $m$). Substitutions can be extended to tuples in the expected way, i.e., pairwise replacement. 

Parameterization refers to abstraction and application, and sometimes parameterization and abstraction are used interchangeably. 
Intuitively the process abstraction $\lrangle{X}P $ (respectively name abstraction $\lrangle{x}P$) \emph{abstracts} in $P$  the process variable $X$ (respectively name variable $x$), which is supposed to be instantiated by a concrete process $Q$ (respectively name $d$) in the application $(\lrangle{X}P)\lrangle{Q}$ (respectively $(\lrangle{x}P)\lrangle{d}$); then in turn the application gives rise to an applied form $P\hosub{P}{X}$ (respectively $P\fosub{d}{x}$). The constructs of abstract and application stem from the counterpart in the lambda-calculus, 
and somehow extend the domain of the lambda-calculus to a concurrent setting.
{To ensure correct use of abstraction and application, a type system was designed by Sangiorgi in his seminal thesis \cite{San92}. The typing rules in the type system effect to exclude badly formed expressions, such as $(\lrangle{x}P)\lrangle{A}$ and $\lrangle{A}P$ in which $A$ is a (non-variable) term, $P\para \lrangle{X}Q$ (dangling abstraction), and so on.  
That type system is important but not essential for our work here, so we do not present it and always assume that terms are well-formed subject to typing; interested readers can refer to \cite{San92, SW01a} and reference thereof  for more details. 
}
Term $P_1\para P_2\para \cdots \para P_k$ is abbreviated as $\Pi_{i{=}1}^{k}P_i$. We also have some CCS-like operations defined as follows: $a.P\DEF a(X).P$ where $X\notin \pv{P}$; $a\DEF a.0$; $\overline{a} \DEF \overline{a}0$.
A context $C[\cdot]$ is an expression with some sub-expression replaced by the hole $[\cdot]$, and $C[A]$ means substituting the hole with $A$. 


%


\iftoggle{mNotsON}{%
} {%
\TODO
\TODOM{Notice that the input rule does not instantiate (late style) and that the part on parameterization can be adapted and inherited}
\TODOM{Explain and convention...}
}


\vspace*{.1cm}
\noindent\textbf{Semantics}~~ We denote by $\SCongru$ the standard structural congruence extended by the rules for application, i.e., the smallest congruence meeting the following laws among which the last two formulate the application. 
\annotate{
\[
\begin{array}{lll}
(P\para Q)\para R \SCongru P\para (Q\para R) \quad & 
P\para Q \SCongru Q\para P \quad & 
P \para 0 \SCongru P \\
(\lrangle{X}P)\lrangle{Q} \SCongru P\hosub{Q}{X} \quad & 
(\lrangle{x}P)\lrangle{m} \SCongru P\hosub{m}{x} &
\end{array}
\]
}
\vspace*{.15cm}
$(P\para Q)\para R \SCongru P\para (Q\para R), \quad  
P\para Q \SCongru Q\para P, \quad 
P \para 0 \SCongru P, \quad
(\lrangle{X}P)\lrangle{Q} \SCongru P\hosub{Q}{X}, \quad  
(\lrangle{x}P)\lrangle{m} \SCongru P\hosub{m}{x}$

Calculus \HOmp\ has the following operational semantics on open terms, with symmetric rules skipped. In the third rule, we assume $\bpv{\lambda}\,\cap\, \fpv{Q}=\emptyset$.
\annotate{Since the semantic rules are based on open terms, a transmitted term may contain free variables.} 
\[
\begin{array}{l}
\infer{m(X).P \st{m(X)} P}{} \quad 
\infer{\overline{m}Q \st{\overline{m}(Q)} 0}{}  \quad 
\infer{P\para Q\st{\lambda} P'\para Q}{P \st{\lambda} P'} \quad 
\infer{P\para Q \st{\tau} P'\para Q'\hosub{A}{X}}{P\st{\overline{m}(A)} P' \quad Q\st{m(X)} Q'} \quad
\infer{Q\st{\lambda} Q'}{Q\SCongru P\quad P\st{\lambda} P'\quad P'\SCongru Q'} \\
\end{array}
\]

The semantics grant a term three kinds of actions: input $P \st{a(X)} P'$ means that $P$ can receive a term on channel $a$ to replace the variable $X$ acting as a place-holder in $P$ (here we have a late instantiation style); output $P\st{\overline{a}(Q)} P'$ means that $P$ can send a term $Q$ (which could be an abstraction) on channel $a$ in an asynchronous fashion; interaction $P\st{\tau} P'$ means that $P$ makes a communication of some term between concurrent components. Actions are ranged over by $\alpha,\lambda$.
Operations $\fpv{\cdot}$, $\bpv{\cdot}$, $\pv{\cdot}$, $\fnv{\cdot}$, $\bnv{\cdot}$, $\nv{\cdot}$, $\nc{\cdot}$, $\n{\cdot}$ and also substitutions can be extended to actions in the expected way accordingly. We sometimes write $P\st{\lambda} \cdot$ to represent the transition $P\st{\lambda} P'$ for some $P'$ if $P'$ is not important.
Modelling application as part of the structural congruence follows the line of reduction in lambda-calculus, though there are other options (see \cite{San92,SW01a}).
Thus up-to $\SCongru$, a term can be somehow turned into an equivalent one by applying applications as many times as possible, ending up with a term containing only those application of the form $X\lrangle{A}$. 
\annotate{For the sake of convenience, we often take the convention that every term satisfies this form as described, unless otherwise stated. }
As in \cite{San92}, we ensure that applications (substitutions) are bound to end (i.e., normalized), so as to avoid $\Omega$-like terms such as $O\lrangle{O}$ in which $O\DEF \lrangle{X}(X\lrangle{X})$. Said another way, in the sense of order, we focus on abstractions with finite order, not $\omega$ order. See \cite{San92,SW01a} for more discussion about this. 
We further notice that if infinite application were to be admitted (though this is a bit strange), then essentially one would retrieve replication, e.g., $!P\DEF O'\lrangle{O'}$ in which $O'\DEF \lrangle{X}(P\para X\lrangle{X})$. This would probably lead to a drastically different situation, which we do not tackle in this work. 
\xxyrmcolor{
Before moving on, we give an example to illustrate the modelling capability of \HOmp. We define two processes $P$ and $Q$ executing a simple protocol, making good use of the parameterization. 
\[P\DEF \overline{a}A \para b(X).(X\lrangle{B} \para O), \qquad Q\DEF a(X).(X\lrangle{c} \para c(Y).R), \qquad A\DEF \lrangle{x}(\overline{b}[\lrangle{Z}\overline{x}Z])
\]
\[
\begin{array}{lcll}
P\para Q &\st{\tau}\SCongru& b(X).(X\lrangle{B} \para O) \para A\lrangle{c} \para c(Y).R &\,\SCongru\, b(X).(X\lrangle{B} \para O) \para \overline{b}[\lrangle{Z}\overline{c}Z] \para c(Y).R \\
&\st{\tau}\SCongru& (\lrangle{Z}\overline{c}Z)\lrangle{B} \para O \para c(Y).R &\,\SCongru\, \overline{c}B \para O \para c(Y).R \\
&\st{\tau}\SCongru& R\hosub{B}{Y} \para O &
\end{array}
\]
The protocol goes as follows: (1) $P$ sends $Q$ an abstraction $A$ over channel $a$ (which is agreed upon beforehand); (2) $Q$ instantiates the name abstraction carried by $A$ with a name $c$ chosen by $Q$ alone (not necessarily negotiated with $P$ before starting the protocol); (3) Part of the code of $A$, i.e., $\lrangle{Z}\overline{c}Z$ is sent back to $P$ over channel $b$ chosen by $P$ alone; (4) Process $B$, e.g., some computational resource or data, is sent to $Q$ over channel $c$, so as to be used in $R$. In the entire protocol, $P$ and $Q$ only agree on the channel name $a$, and initially do not disclose on which channel the resource is to be transmitted. 
}

\annotate{
\textit{\tiny\bc{ \fbox{NOTICE}:
If one does not like the last semantic rule (**) using the structural congruence, it can be replaced with the original rule presented in Sangiorgi's thesis \cite{San92}. That is, one uses the following two rules for name and process abstractions respectively.
\[
\begin{array}{ll}
\infer{(\lrangle{x}P)\lrangle{m} \st{\lambda} P'}{P\fosub{m}{x} \st{\lambda} P'}
\qquad &
\infer{(\lrangle{X}P)\lrangle{Q} \st{\lambda} P'}{P\hosub{Q}{X} \st{\lambda} P'}
\end{array}
\]
Together with these rules, the following lemma helps to achieve the same effect of formulation as the one with the last LTS above. 
\begin{lemma}\label{l:struc_congru_prop}
If $P\st{\lambda} P'$ and $P\SCongru Q$, then there exists $Q'$ such that $Q\st{\lambda} Q'$ with $P'\SCongru Q'$.
\end{lemma}
\begin{proof}
The proof is an induction on the derivation of $P \SCongru Q$ and a case analysis on the action $P\st{\lambda} P'$. 
\end{proof}
}}
}
%
Below we give the notion of ``guarded'' and some relevant properties.
\begin{definition}\label{def:guarded} 
A variable $X$ \annotate{(or $x$) }
is guarded in $P$ 
\annotate{(\xxa{\scriptsize up-to $\SCongru$?? not necessary by convention, i.e., exhausted applications}) }
if $X$ \annotate{(or $x$)} 
merely occurs in the following two situations. 
(1) $X$ \annotate{(or $x$)} 
occurs in $P$'s subexpressions of the form $m(Y).P'$ (in which $Y$ could be the same as $X$),
  or $Y\lrangle{P'}$ (in which $Y$ is not $X$).  
(2) $X$ \annotate{(or $x$)} 
occurs free in $P$'s subexpressions of the form $\overline{m}P'$. ~
A term $P$ is guarded if any free variable of it is guarded.
\end{definition}

In what follows, we have the abbreviations: $\mtriggerName \DEF \mtrigger$, $\mtriggerDName \DEF \mtriggerD$, $\mtriggerDdName \DEF \mtriggerDd$. The proofs of the coming two lemmas are by transition induction.
\begin{lemma}\label{l:lts_prop1}
We have the following transition properties. \\ 
\textbf{(1)} If $P\st{\lambda}P'$, then $P\hosub{R}{X} \st{\lambda\hosub{R}{X}} P'\hosub{R}{X}$ for every $R$ with $\fpv{R}\cap (\pv{P}\cup \pv{\lambda}\cup \{X\})= \emptyset$. \\
\textbf{(2)} If $P\hosub{R}{X} \st{\lambda'} P_1$ with $X$ guarded in $P$ and $\fpv{R}\cap (\pv{P}\cup \{X\})= \emptyset$, then $P\st{\lambda} P'$, $P_1 \SCongru P'\hosub{R}{X}$, and $\lambda'$ is $\lambda\hosub{R}{X}$ with $\fpv{R}\cap \pv{\lambda} = \emptyset$. \\
\textbf{(3)} If $P\hosub{\mtriggerName}{X} \st{\lambda'} P_1$ with $m$ fresh and not in $\lambda'$, then  $P\st{\lambda} P'$, $P_1 \SCongru P'\hosub{\mtriggerName}{X}$, and $\lambda'$ is $\lambda\hosub{\mtriggerName}{X}$. \\
\textbf{(4)} If $P\hosub{\mtriggerDName}{X} \st{\lambda'} P_1$ with $m$ fresh and not in $\lambda'$, then  $P\st{\lambda} P'$, $P_1 \SCongru P'\hosub{\mtriggerDName}{X}$, and $\lambda'$ is $\lambda\hosub{\mtriggerDName}{X}$. \\
\textbf{(5)} If $P\hosub{\mtriggerDdName}{X} \st{\lambda'} P_1$ with $m$ fresh and not in $\lambda'$, then  $P\st{\lambda} P'$, $P_1 \SCongru P'\hosub{\mtriggerDdName}{X}$, and $\lambda'$ is $\lambda\hosub{\mtriggerDdName}{X}$. \\
\textbf{(6)} If $P\st{\lambda}P'$, then $P\fosub{g}{m} \st{\lambda\fosub{g}{m}} P'\fosub{g}{m}$.\\ 
\textbf{(7)} If $P\fosub{g}{m} \st{\lambda'} P_1$ and $\lambda'$ is not $\tau$, then $P\st{\lambda} P'$ in which $\lambda'$ is $\lambda\fosub{g}{m}$, and $P_1\SCongru P'\fosub{g}{m}$.\\
\textbf{(8)} If $P\fosub{g}{m} \st{\tau} P_1$, then there are several possibilities:
\textbf{(a)} $P\st{\tau} P'$ and $P_1\SCongru P'\fosub{g}{m}$.
\textbf{(b)} $P\st{\overline{m}A} \cdot$ and $P\st{g(Y)} \cdot$. That is, $P\SCongru \overline{m}A \para g(Y).P_2 \para P_3$, and $P_1\SCongru (P_2\hosub{A}{Y}\para P_3)\fosub{g}{m}$.
\textbf{(c)} $P\st{\overline{g}A} \cdot$ and $P\st{m(Y)} \cdot$. That is, $P\SCongru \overline{g}A \para m(Y).P_2 \para P_3$, and $P_1\SCongru (P_2\hosub{A}{Y}\para P_3)\fosub{g}{m}$.
\end{lemma}

\begin{lemma}\label{l:lts_prop2}
Assume that $P$ is a term and $X$ is a process variable. There are $P'$ in which $X$ is guarded and natural number $k\geqslant 0$ such that one of the following cases is true. 
\textbf{(1)} $P \SCongru P'\para \Pi^{k}_{i{=}1}X$, and $P\hosub{R}{X} \SCongru P'\hosub{R}{X} \para \Pi^{k}_{i{=}1}R$ for every $R$. 
~~\textbf{(2)} $P \SCongru P'\para \Pi^{k}_{i{=}1}X\lrangle{A_i}$, and $P\hosub{R}{X} \SCongru P'\hosub{R}{X} \para \Pi^{k}_{i{=}1}R\lrangle{A_i\hosub{R}{X}}$ for every $R$. 
~~\textbf{(3)} $P \SCongru P'\para \Pi^{k}_{i{=}1}X\lrangle{m_i}$, and $P\hosub{R}{X} \SCongru P'\hosub{R}{X} \para \Pi^{k}_{i{=}1}R\lrangle{m_i}$ for every $R$. 
\end{lemma}




\vspace*{.1cm}
\noindent\textbf{The strong bisimilarities}~~ 
In the following, we first present a provably decidable strong bisimilarity, named strong HO-IO bisimilarity. Then we go head to define the various strong bisimilarities, including the strong context bisimilarity and other strong bisimilarities of concern. These strong bisimilarities turn out to be equal. 
\annotate{Specifically, we shall define a bundle of bisimilarities on \HOmp. Some of them are on top of open processes and others on closed processes. These bisimilarities can be extended to abstractions (if it is initially not defined on abstractions) and open processes (if it is initially defined on closed processes) in the usual way, i.e., imposing closure under instantiation (substitution) of the corresponding variables. 
}

\vspace*{.1cm}
\noindent\textbf{Strong HO-IO bisimilarity}~~ 
We define a bisimulation called strong HO-IO bisimulation, with the corresponding equality called strong HO-IO bisimilarity. As will be seen, the most desirable properties we want from this bisimilarity is that it is decidable.  
{The definition needs to take into account the abstractions, because the terms transmitted to be compared may be abstractions.}
Jumping ahead, the other strong bisimulations to be defined also have this requirement for abstractions.

\begin{definition}[Strong HO-IO bisimilarity]\label{d:shoio_bisi}
A symmetric binary relation $\R$ over \HOmp\ terms is a strong HO-IO bisimulation, if whenever $P\,\R \, Q$ the following properties hold. 

\noindent\textbf{(1)} If $P$ is a non-abstraction, then so is $Q$. 
\noindent~\textbf{(2)} If $P$ is a process-abstraction $\lrangle{Y}A$, then $Q$ is a process-abstraction $\lrangle{Y}B$, and $A \,\R\, B$. 
\noindent~\textbf{(3)} If $P$ is a name-abstraction $\lrangle{y}A$, then $Q$ is a name-abstraction $\lrangle{y}B$, and $A \,\R\, B$.
\noindent~\textbf{(4)} If $P \st{\overline{a}A} P'$, then $Q\st{\overline{a}B} Q'$ with $A \,\R\, B$ and $P' \,\R\, Q'$. 
\noindent~\textbf{(5)} If $P \st{a(X)} P'$, then $Q \st{a(X)} Q'$ and $P' \,\R\, Q'$. 
\noindent~\textbf{(6)} If $P \SE X \para P'$, then $Q \SE X \para Q'$ and $P' \,\R\, Q'$. 
   \tdup{Is it NECESSAREY to use the following framed variant?? \rc{MAYBE YES; will see!}}
   \tdup{ \item[] \fbox{If $P \SE X\lrangle{A} \para P'$, then $Q \SE X\lrangle{B} \para Q'$, and $A \,\R\, B$ and $P' \,\R\, Q'$.}}
\noindent~\textbf{(7)} If $P \SE X\lrangle{{A}} \para P'$, then $Q \SE X\lrangle{{B}} \para Q'$, and $A \,\R\, B$ and $P' \,\R\, Q'$. 
  \tdup{Is it NECESSAREY to use instead  the following framed variant?? \rc{?; will see!}}
\tdup{ \item[] \fbox{If $P \SE X\lrangle{\xxy{A}} \para P'$, then $Q \SE X\lrangle{\xxy{A}} \para Q'$ and $P' \,\R\, Q'$.}}
\noindent~\textbf{(8)} If $P \SE X\lrangle{d} \para P'$, then $Q \SE X\lrangle{d} \para Q'$ and $P' \,\R\, Q'$.  
  \tdup{Is it POSSIBLE to use instead  the following framed variant?? \rc{?;Probably NOT! NOT making much sense!}}
\tdup{ \item[] \fbox{If $P \SE X\lrangle{\xxy{d}} \para P'$, then $Q \SE X\lrangle{\xxy{e}} \para Q'$ and $P' \,\R\, Q'$.}}
The strong HO-IO bisimilarity, notation \SHOIOB, is the largest strong HO-IO bisimulation. 
\end{definition}

\vspace*{.1cm}
\noindent\textbf{Strong HO bisimilarity}~~ 
The concept of strong HO bisimilarity is due to Thomsen \cite{Tho93}.
\begin{definition}[Strong HO bisimilarity]\label{d:sho_bisi}
A symmetric binary relation $\R$ over closed \HOmp\ terms is a strong HO bisimulation, if whenever $P\,\R \, Q$ the following properties hold. 

\noindent\textbf{(1)} If $P$ is a non-abstraction, then so is $Q$. 
\noindent~\textbf{(2)} If $P$ is a process-abstraction $\lrangle{Y}P'$, then $Q$ is a process-abstraction $\lrangle{Y}Q'$, and $P'\hosub{A}{Y} \,\R\, Q'\hosub{A}{Y}$ for every closed $A$. 
\noindent~\textbf{(3)} If $P$ is a name-abstraction $\lrangle{y}A$, then $Q$ is a name-abstraction $\lrangle{y}B$, and $A \,\R\, B$.
\noindent~\textbf{(4)} If $P \st{\overline{a}A} P'$, then $Q\st{\overline{a}B} Q'$ with $A \,\R\, B$ and $P' \,\R\, Q'$. 
\noindent~\textbf{(5)} If $P \st{a(X)} P'$, then $Q \st{a(X)} Q'$ and for every closed $A$, it holds that $P'\hosub{A}{X} \,\R\, Q'\hosub{A}{X}$. 
\noindent~\textbf{(6)} If $P \st{\tau} P'$, then $Q \st{\tau} Q'$ and $P' \,\R\, Q'$. ~~
The strong HO bisimilarity, notation \SHOB, is the largest strong HO bisimulation. 
\end{definition}

\vspace*{.1cm}
\noindent\textbf{Strong context bisimilarity}~~ 
We denote by $E(X)$ a process $E$ possibly with $X$ appearing free in it, i.e., $\fpv{E} \subseteq \{X\}$. Accordingly, $E(A)$ denotes $E(X)\hosub{A}{X}$. 
As a nearly standard version of the bisimilarity for higher-order processes, the context bisimulation was proposed by Sangiorgi \cite{San92}.
\begin{definition}[Strong context bisimilarity]\label{d:sctx_bisi}
A symmetric binary relation $\R$ over closed \HOmp\ terms 
is a strong context bisimulation, if whenever $P\,\R \, Q$ the following properties hold. 
{
\noindent\textbf{(1)} If $P$ is a non-abstraction, then so is $Q$. 
\noindent~\textbf{(2)} If $P$ is a process-abstraction $\lrangle{Y}P'$, then $Q$ is a process-abstraction $\lrangle{Y}Q'$, and $P'\hosub{A}{Y} \,\R\, Q'\hosub{A}{Y}$ for every closed $A$. 
\noindent~\textbf{(3)} If $P$ is a name-abstraction $\lrangle{y}A$, then $Q$ is a name-abstraction $\lrangle{y}B$, and $A \,\R\, B$.
}
\tdup{\scriptsize
\item[---]  \xxy{Do we need the above few clauses for abstractions? That is, may we need to modify the def of strong context bisimulation to admit comparison of abstractions? \rc{Seems YES; TO see... } } ...\\
\xxy{In the definition of the strong context bisimulation, these requirements about abstractions appear not necessary in terms of output, because the pair of terms (which can be abstractions) being outputted are never compared directly without any contexts in the situation of the context bisimulation. 
However, they appear necessary for input, because the matching of an input can produce the comparison of two terms that are abstraction (this contrasts the situation where processes rather than general terms are considered in the definition of the strong context bisimulation). For example, $a(X).\lrangle{Y}0$ and $a(X).\lrangle{Z}0$ will ask one to compare $\lrangle{Y}0$ and $\lrangle{Z}0$. Though this may be simplified, we put these requirements here in the definition for the sake of safety and convenience.
}
}
\noindent~\textbf{(4)} If $P \st{a(X)} P'$, then $Q \st{a(X)} Q'$ and for every closed $A$, it holds that $P'\hosub{A}{X} \,\R\, Q'\hosub{A}{X}$. 
\noindent~\textbf{(5)} If $P \st{\overline{a}A} P'$ in which $A$ is a non-abstraction, process-abstraction , or name-abstraction, then $Q \st{\overline{a}B} Q'$ for some $B$ that is respectively a non-abstraction, process-abstraction , or name-abstraction, and for every $E(X)$, it holds that $E(A)\para P' \,\R\, E(B)\para Q'$. 
\noindent~\textbf{(6)} If $P \st{\tau} P'$, then $Q \st{\tau} Q'$ and $P' \,\R\, Q'$.~~ 
The strong context bisimilarity, notation \SCTXB, is the largest strong context bisimulation. 
\end{definition}

We note that $\SCTXB$ can be extended to open process terms, similar for $\SHOB$. That is, for open terms $P$ and $Q$ with $\fpv{P,Q}=\ve{X}$, 
$
P \,\SCTXB\, Q \mbox{ if and only if } P\hosub{\ve{R}}{\ve{X}} \,\SCTXB\, Q\hosub{\ve{R}}{\ve{X}}  \mbox{ for any closed } \ve{R} .
$

\iftoggle{mNotsON}{%
} {%
\TODOM{expand...discussion...properties of strong context bisimilarity}
\TODOM{notice extension to abstraction and open processes...}
}

\vspace*{.1cm}
\noindent\textbf{Strong normal bisimilarity}~~ 
Higher-order process expressions here may be parameterized over processes themselves or names. 
Accordingly, abstractions can be transmitted in communications, and thus process variables have three types: non-abstraction, process-abstraction, and name-abstraction. 
We refer the reader to \cite{San92} for the detailed formalization of types. 
To cater for our need, knowing which of the three types a process variable belongs to is sufficient for our work. For convenience, we may simply say that a process variable is a non-abstraction, process-abstraction, or name-abstraction. 

Before presenting the definition of the strong normal bisimilarity, we give the definition of triggers: $\mtriggerName \DEF \mtrigger$, $\mtriggerDName \DEF \mtriggerD$, $\mtriggerDdName \DEF \mtriggerDd$. These triggers correspond to the three types of process variables represented above, and will be used to handle abstractions bound to instantiate these process variables. The concept of triggers was proposed by Sangiorgi and plays a prevalent role in the manipulation of higher-order processes; see \cite{San92}\cite{Xu20}. 
\xxyrmcolor{
 We stress that the design of the normal bisimulation in this work requires new forms of triggers due to the  presence of parameterization. The work in \cite{LPSS10a} only needs the simplest form of triggers acting as synchronizers sending handshaking signals, i.e., $\mtrigger$. However in contrast, in the setting of parameterization, triggers should bear the responsibility of relocating the parameters for an abstraction. This design is non-trivial in general, and we harness the results in the previous work \cite{Xu20} to devise different forms of triggers used by the parameterization.
}
It is not hard to prove that the strong normal bisimilarity is a congruence \cite{San92}.
\begin{definition}[Strong normal bisimilarity]\label{d:snr_bisi}
A symmetric binary relation $\R$ over closed \HOmp\ terms 
 is a strong normal bisimulation, if whenever $P\,\R \, Q$ the following properties hold. 

{
\noindent\textbf{(1)} If $P$ is a non-abstraction, then so is $Q$. 
\noindent~\textbf{(2)} If $P$ is a process-abstraction $\lrangle{Y}P'$, then $Q$ is a process-abstraction $\lrangle{Y}Q'$, and  for every closed $A$ it holds for fresh $m$ that: \\
\noindent~\textbf{(a)} $P'\hosub{\mtriggerName}{Y} \,\R\, Q'\hosub{\mtriggerName}{Y}$, if $Y$ is a non-abstraction. 
\noindent~\textbf{(b)} $P'\hosub{\mtriggerDName}{Y} \,\R\, Q'\hosub{\mtriggerDName}{Y}$, if $Y$ is a process-abstraction. 
\noindent~\textbf{(c)} $P'\hosub{\mtriggerDdName}{Y} \,\R\, Q'\hosub{\mtriggerDdName}{Y}$, if $Y$ is a name-abstraction. 

\noindent\textbf{(3)} If $P$ is a name-abstraction $\lrangle{y}A$, then $Q$ is a name-abstraction $\lrangle{y}B$, and $A \,\R\, B$.
}
\tdup{\scriptsize
\item[---]  \xxy{Do we need the clauses for abstractions as the corresponding \#1,\#2,\#3 ones in the definition of strong context bisimulation (Definition \ref{d:sctx_bisi})?  \rc{Seems YES; TO see (synchronously with that in the def of strong context bisimulation)...}} ...
}

\noindent\textbf{(4)} If $P \st{a(X)} P'$, then $Q \st{a(X)} Q'$ and for every closed $A$, it holds for fresh $m$ that: \\
\noindent~\textbf{(a)} $P'\hosub{\mtriggerName}{X} \,\R\, Q'\hosub{\mtriggerName}{X}$, if $X$ is a non-abstraction. 
\noindent~\textbf{(b)} $P'\hosub{\mtriggerDName}{X} \,\R\, Q'\hosub{\mtriggerDName}{X}$, if $X$ is a process-abstraction. 
\noindent~\textbf{(c)} $P'\hosub{\mtriggerDdName}{X} \,\R\, Q'\hosub{\mtriggerDdName}{X}$, if $X$ is a name-abstraction. 

\noindent~\textbf{(5)} If $P \st{\overline{a}A} P'$, there are three possibilities: 
\noindent~\textbf{(a)} If $A$ is not an abstraction, then $Q\st{\overline{a}B} Q'$ for non-abstraction $B$, and it holds for fresh $m$ that $m.A\para P' \,\R\, m.B\para Q'$.
\noindent~\textbf{(b)} If $A$ is a process-abstraction $\lrangle{Y}A_1$, then $Q\st{\overline{a}B} Q'$ for process-abstraction $B$ that is $\lrangle{Y}B_1$,  and it holds for fresh $m$ that $m(Z).A\lrangle{Z}\para P' \,\R\, m(Z).B\lrangle{Z}\para Q'$.
\noindent~\textbf{(c)} If $A$ is a name-abstraction $\lrangle{y}A_1$, then $Q\st{\overline{a}B} Q'$ for name-abstraction $B$ that is $\lrangle{y}B_1$,  and it holds for fresh $m$ that $m(Z).Z\lrangle{A}\para P' \,\R\, m(Z).Z\lrangle{B}\para Q'$.

\noindent~\textbf{(6)} If $P \st{\tau} P'$, then $Q \st{\tau} Q'$ and $P' \,\R\, Q'$.~~
The strong normal bisimilarity, notation \SNRB, is the largest strong normal bisimulation. 
\end{definition}

We can also extend $\SNRB$ to open terms. 
\annotate{Here we should explicitly discriminate between different types of process variables. }
For open terms $P$ and $Q$ with $\fpv{P, Q} = \{\ve{X_1}, \ve{X_2}, \ve{X_3}\}$, 
$
P \,\SNRB\, Q ~~
\mbox{ if and only if } ~~
P
\hosub{\ve{\genTrigger[m_1]{}}}{\ve{X_1}} 
\hosub{\ve{\genTrigger[m_2]{D}}}{\ve{X_2}} 
\hosub{\ve{\genTrigger[m_3]{D,d}}}{\ve{X_3}} 
\,\SNRB\,
Q
\hosub{\ve{\genTrigger[m_1]{}}}{\ve{X_1}} 
\hosub{\ve{\genTrigger[m_2]{D}}}{\ve{X_2}} 
\hosub{\ve{\genTrigger[m_3]{D,d}}}{\ve{X_3}}
$,  
where 
each variable in $\ve{X_1}$, $\ve{X_2}$ and $\ve{X_3}$ is respectively a  non-abstraction, process-abstraction and name-abstraction, and is replaced with the corresponding trigger for that variable type. The corresponding tuples of triggers are respectively denoted by $\ve{\genTrigger[m_1]{}}$, $\ve{\genTrigger[m_2]{D}}$ and $\ve{\genTrigger[m_3]{D,d}}$, where the names of all the triggers are fresh. 
\tdup{\scriptsize This extension is also applicable to abstractions. MAYBE UNnecessary because it is already embedded in the definition. }

\iftoggle{mNotsON}{%
} {%
\TODOM{expand...discussion...properties of strong normal bisimilarity}
\TODOM{notice extension to abstraction and open processes...}
}

\vspace*{.1cm}
\noindent\textbf{Open strong normal bisimilarity}~~ 
The following bisimilarity is a variant of the strong normal bisimilarity on open terms. It is basically an extension of the same bisimilarity in \cite{LPSS10a}.
\begin{definition}[Open strong normal bisimilarity]\label{d:opensnr_bisi}
A symmetric binary relation $\R$ over \HOmp\ terms is an open strong normal bisimulation, if whenever $P\,\R \, Q$ the following properties hold. 

\noindent\textbf{(1)} If $P$ is a non-abstraction, then so is $Q$. 
\noindent~\textbf{(2)} If $P$ is a process-abstraction $\lrangle{Y}A$, then $Q$ is a process-abstraction $\lrangle{Y}B$, and $A \,\R\, B$. 
\noindent~\textbf{(3)} If $P$ is a name-abstraction $\lrangle{y}A$, then $Q$ is a name-abstraction $\lrangle{y}B$, and $A \,\R\, B$.
\noindent~\textbf{(4)} If  $P \st{\overline{a}A} P'$ or $P \st{\tau} P'$, then $Q$ matches $P$ in the same way as in strong normal bisimilarity. 
\noindent~\textbf{(5)} If $P \st{a(X)} P'$, then $Q \st{a(X)} Q'$ and $P' \,\R\, Q'$. 
\noindent~\textbf{(6)} If $P \SE X \para P'$, then $Q \SE X \para Q'$ and $P' \,\R\, Q'$. 

{ 
\noindent\textbf{(7)} If $P \SE X\lrangle{\xxyrmcolor{A}} \para P'$, then $Q \SE X\lrangle{\xxyrmcolor{B}} \para Q'$ and moreover the following is valid for fresh $m$.\\
\noindent~\textbf{(a)} If $A$ is not an abstraction, then so is $B$ and $m.A\para P' \,\R\, m.B\para Q'$.
\noindent~\textbf{(b)} If $A$ is a process-abstraction, then so is $B$ and  $m(Z).A\lrangle{Z}\para P' \,\R\, m(Z).B\lrangle{Z}\para Q'$.
\noindent~\textbf{(c)} If $A$ is a name-abstraction, then so is $B$ and $m(Z).Z\lrangle{A}\para P' \,\R\, m(Z).Z\lrangle{B}\para Q'$.

}

  \tdup{Is it NECESSAREY to use instead the following framed variant?? \rc{?; will see!}}
\tdup{\item[] \fbox{If $P \SE X\lrangle{\rc{A}} \para P'$, then $Q \SE X\lrangle{\rc{B}} \para Q'$, and $A \,\R\, B$ and $P' \,\R\, Q'$. }}

  \tdup{OR Is it NECESSAREY to use instead the following another variant?? \rc{?; seems too demanding!}}
\tdup{ \item[] \fbox{If $P \SE X\lrangle{\xxy{A}} \para P'$, then $Q \SE X\lrangle{\xxy{A}} \para Q'$ and $P' \,\R\, Q'$.}}

\noindent~\textbf{(8)} If $P \SE X\lrangle{d} \para P'$, then $Q \SE X\lrangle{d} \para Q'$ and $P' \,\R\, Q'$.  
  \tdup{Is it POSSIBLE to use instead  the following framed variant?? \rc{?;Probably NOT! NOT making much sense!}}
  \tdup{ \item[] \fbox{If $P \SE X\lrangle{\xxy{d}} \para P'$, then $Q \SE X\lrangle{\xxy{e}} \para Q'$ and $P' \,\R\, Q'$.}}

\noindent The open strong normal bisimilarity, notation \OSNRB, is the largest open strong normal bisimulation. 
\end{definition}

\iftoggle{mNotsON}{%
} {%
\TODOM{expand...discussion...properties of open strong HO-IO bisimilarity}
}

\vspace*{-.4cm}
\section{Deciding the strong bisimilarity for \HOmp}\label{s:dec_big_d_small_d}


In this section, we first establish the decidability of the strong HO-IO bisimilarity. 
This is the cornerstone of the decidability for other bisimilarities. Then we discuss the relationship between the strong bisimilarities, and eventually obtain the coincidence between them. As such, all of the strong bisimilarities are decidable.

\vspace*{-.4cm}
\subsection{The decidability and properties of \SHOIOB}


%
To facilitate discussion on decidability, we need a metric of the syntactical structure of a term. 
\begin{definition}[Depth of a term]\label{def:depth}
The depth $\pdepth{P}$ of a term $P$ is a mapping from terms to natural numbers defined as follows. \vspace*{.1cm}\\
\annotate{
\[
\begin{array}{ll}
P\qquad\qquad & \pdepth{P} \\\hline 
0 & 0 \\
X & 1 \\
m(X).P_1 & \pdepth{P_1} + 1 \\
\overline{m}(P_1) & \pdepth{P_1} + 1 \\
P_1\para P_2 & \pdepth{P_1} + \pdepth{P_2} \\
\lrangle{X}P_1 & \pdepth{P_1} + 1 \\
X\lrangle{P_1} & \pdepth{P_1} + 1 \\
P_1\lrangle{P_2} & \pdepth{P_3\hosub{P_2}{Y}} \qquad\, \mbox{where $P_1$ is $\lrangle{Y}P_3$} \\
\lrangle{x}P_1 & \pdepth{P_1} + 1 \\
X\lrangle{n} & 1 \\
P_1\lrangle{n} & \pdepth{P_3\fosub{n}{y}} \qquad\quad \mbox{where $P_1$ is $\lrangle{y}P_3$} \\
\end{array}
\]
}
$
\begin{array}{l}
\pdepth{0} = 0,~
\pdepth{X} = 1,~
\pdepth{m(X).P_1} = \pdepth{P_1} + 1,~
\pdepth{\overline{m}(P_1)} = \pdepth{P_1} + 1,~\\
\pdepth{P_1\para P_2} = \pdepth{P_1} + \pdepth{P_2},~
\pdepth{\lrangle{X}P_1} = \pdepth{P_1} + 1,~
\pdepth{X\lrangle{P_1}} = \pdepth{P_1} + 1,~\\
\pdepth{P_1\lrangle{P_2}} = \pdepth{P_3\hosub{P_2}{Y}} ~ \mbox{(where $P_1$ is $\lrangle{Y}P_3$)},~
\pdepth{\lrangle{x}P_1} = \pdepth{P_1} + 1,~\\
\pdepth{X\lrangle{n}} = 1,~
\pdepth{P_1\lrangle{n}} = \pdepth{P_3\fosub{n}{y}} ~ \mbox{(where $P_1$ is $\lrangle{y}P_3$)}
\end{array}
$

\end{definition}

An immediate property is that both of $P\SCongru Q$ and $P\SHOIOB Q$ implies $\pdepth{P} = \pdepth{Q}$. The proof of this property is by induction over the depth of $P$. 
\iftoggle{appendixON}{%
\rc{The details are put in Appendix \ref{a:proofs_1}.}
} {%
\xxyrmcolor{The details can be found in \cite{XZ21L}.}
}

\begin{lemma}\label{l:depth_struc_sim}
If $P\SCongru Q$ or $P\SHOIOB Q$, then $\pdepth{P} = \pdepth{Q}$.
\end{lemma}
\annotate{
\begin{lemma}\label{l:depth_sim}
If $P\SHOIOB Q$, then $\pdepth{P} = \pdepth{Q}$.
\end{lemma}
}

%
\vspace*{.1cm}
\noindent\textbf{Strong HO-IO bisimulation  up-to $\SCongru$}~~ 
\annotate{Bisimulation up-to $\SCongru$ (and others equalities as well) is a useful technique to establish bisimulations; see \cite{Mil89, SW01a} for a thorough introduction and discussion. Here for our purpose, we define strong HO-IO bisimulation up-to $\SCongru$. }
Bisimulation up-to $\SCongru$ is a useful technique to establish bisimulations. Its definition is obtained by replacing $\R$ with $\SCongru\R\SCongru$ in every clause of Definition \ref{d:shoio_bisi}. The advantage of the up-to technique is that if \R is a strong HO-IO bisimulation up-to $\SCongru$, then $\R \subseteq \SHOIOB$.
See \cite{Mil89, SW01a} for a thorough introduction and discussion. 

\annotate{ 
\begin{definition}[Strong HO-IO bisimulation up-to $\SCongru$]
A symmetric binary relation $\R$ over (\HOmp) terms is a strong HO-IO bisimulation up-to $\SCongru$, if whenever $P\,\R \, Q$ the following properties hold. 
\begin{enumerate}


\item If $P$ is a non-abstraction, then so is $Q$. 
\item If $P$ is a process-abstraction $\lrangle{Y}A$, then $Q$ is a process-abstraction $\lrangle{Y}B$, and $A \,\SCongru\R\SCongru\, B$. 

\item If $P$ is a name-abstraction $\lrangle{y}A$, then $Q$ is a name-abstraction $\lrangle{y}B$, and $A \,\SCongru\R\SCongru\, B$.


\item If $P \st{\overline{a}(A)} P'$, then $Q\st{\overline{a}(B)} Q'$ with $A \,\SCongru\R\SCongru\, B$ and $P' \,\SCongru\R\SCongru\, Q'$. 
\item If $P \st{a(X)} P'$, then $Q \st{a(X)} Q'$ and $P' \,\SCongru\R\SCongru\, Q'$. 

\item If $P \SCongru X \para P'$, then $Q \SCongru X \para Q'$ and $P' \,\SCongru\R\SCongru\, Q'$. 

\item If $P \SCongru X\lrangle{{A}} \para P'$, then $Q \SCongru X\lrangle{{B}} \para Q'$, and $A \,\SCongru\R\SCongru\, B$ and $P' \,\SCongru\R\SCongru\, Q'$. 
  \tdup{Is it NECESSAREY to use instead  the following framed variant?? \rc{?; will see!}}
\tdup{ \item[] \fbox{If $P \SCongru X\lrangle{\xxy{A}} \para P'$, then $Q \SCongru X\lrangle{\xxy{A}} \para Q'$ and $P' \,\SCongru\R\SCongru\, Q'$.}}

\item If $P \SCongru X\lrangle{d} \para P'$, then $Q \SCongru X\lrangle{d} \para Q'$ and $P' \,\SCongru\R\SCongru\, Q'$.  
  \tdup{Is it POSSIBLE to use instead  the following framed variant?? \rc{?;Probably NOT! NOT making much sense!}}
\tdup{ \item[] \fbox{If $P \SCongru X\lrangle{\xxy{d}} \para P'$, then $Q \SCongru X\lrangle{\xxy{e}} \para Q'$ and $P' \,\SCongru\R\SCongru\, Q'$.}}
\end{enumerate}
\end{definition}
\begin{lemma}\label{l:up-to_correctness}
If \R is a strong HO-IO bisimulation up-to $\SCongru$, then $\R \subseteq \SHOIOB$.
\end{lemma}
\begin{proof}
The proof employs a standard bisimulation-establishing procedure (as in \cite{Mil89, SW01a} for example) to show that the compositional relation $\SCongru\R\SCongru$ is a strong HO-IO bisimulation. 
\end{proof}
}

%


\vspace*{.1cm}
\noindent\textbf{Congruence}~~ 
Through standard state-diagram-chasing argument, one can prove that $\SHOIOB$ is an equivalence relation. 
\annotate{We now show that $\SHOIOB$ is a congruence.}
It is also a congruence, as the follow-up lemma reveals. See \cite{San92, LPSS10a} for a reference of proof; 
\iftoggle{appendixON}{%
\rc{we also provide a proof in Appendix \ref{a:proofs_1}.}
} {%
\xxyrmcolor{we also provide a proof in \cite{XZ21L}.}
}
.
\annotate{  
\begin{lemma}[Equivalence]
On \HOmp\ terms, $\SHOIOB$ is an equivalence relation.
\end{lemma}
\begin{proof}
The proof proceeds by a standard state-diagram-chasing argument, as in \cite{Mil89, SW01a}. 
\end{proof}
}

\begin{lemma}[Congruence]\label{l:congruence}
On \HOmp\ terms, $\SHOIOB$ is congruence. That is, suppose $P$ and $Q$ are \HOmp\  terms, then $P \,\SHOIOB\, Q$ implies: 
\noindent\textbf{(1)} $a(X).P \,\SHOIOB\, a(X).Q$;
\noindent~\textbf{(2)} $\overline{a}(P) \,\SHOIOB\, \overline{a}(Q)$;
\noindent~\textbf{(3)} $P\para R \,\SHOIOB\, Q\para R$;
\noindent~\textbf{(4)} $\lrangle{X}P \,\SHOIOB\, \lrangle{X}Q$;
\noindent~\textbf{(5)} $\lrangle{x}P \,\SHOIOB\, \lrangle{x}Q$; 
\noindent~\textbf{(6)} $Y\lrangle{P} \,\SHOIOB\, Y\lrangle{Q}$.
\end{lemma}

\vspace*{.1cm}
\noindent\textbf{Decidability}~~ 
We now establish the decidability of $\SHOIOB$.  
As a premise, we have the following structural property, whose proof is a simple induction over the semantic rules. 
\begin{lemma}\label{l:struc_prop1}
Suppose $P$ is a \HOmp\ term. Then: 
\noindent\textbf{(1)} If $P\st{\overline{a}(A)} P'$, then $P\SCongru \overline{a}(A) \para P'$.
\noindent~\textbf{(2)} If $P\st{a(X)} P'$, then $P\SCongru a(X).P_1 \para P_2$ and $P'\SCongru P_1 \para P_2$.
\end{lemma}

\begin{lemma}[Decidability]\label{l:decidability_hoio_bisi}
On \HOmp\ terms, $\SHOIOB$ is decidable.
\end{lemma}
\vspace*{-.4cm}
\begin{proof}[Proof of Lemma \ref{l:decidability_hoio_bisi}]
\iftoggle{mNotsON}{%
} {%
\TODOM{in steps...DN}
}

We decide whether $P \SHOIOB Q$ by induction on $\pdepth{P}$. \annotate{Passing all checks would entail $P \SHOIOB Q$. } 

\noindent\textbf{Induction basis.}~
In this case $\pdepth{P}$ is $0$ or $1$. The case $\pdepth{P}$ is $0$, i.e., $P$ is $0$, is trivial because no action is possible from $P$ and it has no free variables. 
If $\pdepth{P}$ is $1$, i.e., $P$ is $X$ or $X\lrangle{d}$, then no action is possible from $P$. One simply checks that $Q$ is also $X$ or $X\lrangle{d}$ respectively. 

\annotate{
\begin{enumerate}
\item $\pdepth{P}$ is $0$, i.e., $P$ is $0$. This case is trivial because no action is possible from $P$ and it has no free variables. 
{\scriptsize There should be nothing visible in $Q$ (up-to $\SCongru$), actions or free variables. }
\item $\pdepth{P}$ is $1$, i.e., $P$ is $X$ or $X\lrangle{d}$. In this case, no action is possible from $P$. One simply checks that $Q$ is also $X$ or $X\lrangle{d}$ respectively. 
\end{enumerate}
}

\noindent\textbf{Induction step.}~
We perform a (finite) check of each clause of $\SHOIOB$ in an inductive way.

\caseHintRM{() (non-abstraction form) ~~}
\noindent\textbf{(1)} If $P$ is a non-abstraction, then check that $Q$ is also a non-abstraction. 

\caseHintRM{() (process-abstraction form) ~~}
\noindent\textbf{(2)} If $P$ is a process-abstraction $\lrangle{X}A$, then check that $Q$ is also a process-abstraction $\lrangle{X}B$ (up-to $\alpha$-conversion), and continue with checking $A \,\SHOIOB\, B$ using induction hypothesis since the depth of $A$ decreases with respect to $P$, i.e., $\pdepth{A} < \pdepth{P}$. 

\caseHintRM{() (name-abstraction form) ~~}
\noindent\textbf{(3)} If $P$ is a name-abstraction $\lrangle{x}A$, then check that $Q$ is also a name-abstraction $\lrangle{x}B$ (up-to $\alpha$-conversion), and continue with checking $A \,\SHOIOB\, B$ using induction hypothesis since the depth of $A$ decreases with respect to $P$, i.e., $\pdepth{A} < \pdepth{P}$.

\caseHintRM{() (Input simulation 1) ~~}  
\noindent\textbf{(4)} If $P \st{a(X)} P'$, we check that $Q \st{a(X)} Q'$ and $P' \SHOIOB Q'$. There might be a few (but finite) possibilities concerning $Q'$. If all such checks fail, then we conclude that $P$ and $Q$ are not strong HO-IO bisimilar. 
For each possible check, we know by Lemma \ref{l:struc_prop1} that $P\SCongru a(X).P_1 \para P_2$ and $P'\SCongru P_1 \para P_2$. Since the depth of the terms decrease, i.e., $\pdepth{P'} < \pdepth{P}$, we use induction hypothesis to continue checking $P'\SHOIOB Q'$. 


\caseHintRM{() (Output simulation) ~~} 
\noindent\textbf{(5)} If $P \st{\overline{a}(A)} P'$, we check that $Q \st{\overline{a}(B)} Q'$ with $A \SHOIOB B$ and $P' \SHOIOB Q'$. There might be a few (but finite) possibilities concerning $B$ and $Q'$. If all such checks fail, then we conclude that $P$ and $Q$ are not strong HO-IO bisimilar. 
For each possible check, we know by Lemma \ref{l:struc_prop1} that $P\SCongru \overline{a}(A) \para P'$.
Since the depth of the terms decrease, i.e., $\pdepth{P'} < \pdepth{P}$ and $\pdepth{A} < \pdepth{P}$, we use induction hypothesis to continue checking $A\SHOIOB B$ and $P'\SHOIOB Q'$.

\caseHintRM{() (Open-form simulation case 1) ~~} 
\noindent\textbf{(6)} If $P \SCongru X \para P'$, 
we check that $Q \SCongru X \para Q'$ and $P' \SHOIOB Q'$. There might be a few (but finite) possibilities concerning $Q'$. If all such checks fail, then we conclude that $P$ and $Q$ are not strong HO-IO bisimilar. 
For each possible check, since the depth of the terms decrease, i.e., $\pdepth{P'} < \pdepth{P}$, we use induction hypothesis to continue checking $P'\SHOIOB Q'$. 

\caseHintRM{() (Open-form simulation case 2) (\rc{similar to 'Open-form simulation case 1' and 'Output simulation'}) ~~} 
\noindent\textbf{(7)} If $P \SCongru X\lrangle{A} \para P'$, 
we check that $Q \SCongru X\lrangle{B} \para Q'$ with $A \SHOIOB B$ and $P' \SHOIOB Q'$. There might be a few (but finite) possibilities concerning $B$ and $Q'$. If all such checks fail, then we conclude that $P$ and $Q$ are not strong HO-IO bisimilar. 
For each possible check, since the depth of the terms decrease, i.e., $\pdepth{P'} < \pdepth{P}$ and $\pdepth{A} < \pdepth{P}$, we use induction hypothesis to continue checking $A\SHOIOB B$ and $P'\SHOIOB Q'$.

\caseHintRM{() (Open-form simulation case 3) (\rc{similar to 'Open-form simulation case 1'}) ~~} 
\noindent\textbf{(8)} If $P \SCongru X\lrangle{d} \para P'$, 
we check that $Q \SCongru X\lrangle{d} \para Q'$ and $P' \SHOIOB Q'$. There might be a few (but finite) possibilities concerning $Q'$. If all such checks fail, then we conclude that $P$ and $Q$ are not strong HO-IO bisimilar. 
For each possible check, since the depth of the terms decrease, i.e., $\pdepth{P'} < \pdepth{P}$, we use induction hypothesis to continue checking $P'\SHOIOB Q'$. 
\end{proof}
\vspace*{-.1cm}


\noindent\textbf{Bisimilarity Preservation}~~ 
To connect the strong HO-IO bisimilarity with the strong context bisimilarity and other bisimilarities, we need some preparation. In what follows, we show that the strong HO-IO bisimilarity preserves substitutions and $\tau$ simulation, as stated in the following two lemmas in a sequel. 
\tdup{Do we need preserving-substitution of names? Seems no, because no name-passing!? But seems technically needed in showing the coincidence between the bisimilarities!? TO SEE if technicality passes!}


\begin{lemma} [Name-substitution-preserving]\label{l:name-sub-preserving}
Assume $P \SHOIOB Q$. Then $P\fosub{g}{m} \SHOIOB Q\hosub{g}{m}$ for all $g,m$.
\end{lemma}


{We note that to keep well-formed, substitutions should not (and are always assumed not to) break the legality of the terms under operation.} 
The following lemma states that \SHOIOB\ is invariant with respect to process substitution.
As mentioned, the proof of this lemma (i.e., the process substitution preserving property) is entirely different from the counterpart in \cite{LPSS10a}, in that we are obliged to conduct an induction on the sizes of the terms because a term in the position of application, say the term $A$ in $X\lrangle{A}$, may introduce extra structures. 
\xxyrmcolor{
 That is, the proof of the preservation of process substitution becomes much more involved due to the process parameterization. In the current setting, a process variable can be instantiated by a process abstraction which is in turn fed with a process from the context. This would give rise to certain circular arguments, so the original proof method of \cite{LPSS10a} no longer works. To work around this difficulty, one has to use induction based approach. This approach somewhat reminds one of the difficulty in proving the congruence properties for higher-order processes.
}
\begin{lemma} [Process-substitution-preserving] \label{l:process-sub-preserving}
Let $P \SHOIOB Q$. Then $P\hosub{R}{X} \SHOIOB Q\hosub{R}{X}$ for all $R,X$.
\end{lemma}

\iftoggle{appendixON}{%
\rc{The proofs of Lemmas \ref{l:name-sub-preserving} and \ref{l:process-sub-preserving} are put in Appendix \ref{a:proofs_1}.}
} {%
\xxyrmcolor{The proofs of Lemmas \ref{l:name-sub-preserving} and \ref{l:process-sub-preserving} are referred to \cite{XZ21L}.}
}
An observation as a corollary from these two lemmas 
is that $\SHOIOB$ is closed under abstraction, both name abstraction and process-abstraction.
\begin{corollary}\label{cor:shoiob-closed-abstr}
Assume $P \SHOIOB Q$. It holds that $\lrangle{X}P \,\SHOIOB\, \lrangle{X}Q$ and $\lrangle{x}P \,\SHOIOB\, \lrangle{x}Q$.
\end{corollary}


The following lemma is important and directly attributed to the open-style nature of \SHOIOB. 
\begin{lemma}[$\tau$-preserving]\label{l:tau-preserving}
Assume $P \SHOIOB Q$. If $P\st{\tau} P'$, then $Q\st{\tau} Q'$ and $P' \SHOIOB Q'$.
\end{lemma}
\vspace*{-.4cm}
\begin{proof}[Proof of Lemma \ref{l:tau-preserving}]

\iftoggle{mNotsON}{%
} {%
\TODOM{in steps...DN}
\TODOM{...\rc{USE Lemma \ref{l:process-sub-preserving}... } DN}
}

The proof is by induction on the derivation of $P\st{\tau} P'$. 
\annotate{There are several cases based on the LTS.} 

\noindent\textbf{(1)} $P\st{\tau} P'$ comes from the interaction of components of $P$. 
\annotate{
That is, 
\[
\begin{array}{l}
P\st{\overline{a}(A)} \cdot \\
P\st{a(X)} \cdot \\
P\st{\tau} P'
\end{array}
\] 
}
That is, $P\st{\overline{a}(A)} \cdot$, $P\st{a(X)} \cdot$, and $P\st{\tau} P'$. 
One can assume $X$ to be fresh as it is bound. So this  can be rewritten as 
\annotate{\[
\begin{array}{l}
P \st{\overline{a}(A)} \cdot \st{a(X)}   P_1 \\
\mbox{ where } P_1\hosub{A}{X}\SCongru  P'
\end{array}
\]
}
$P \st{\overline{a}(A)} \cdot \st{a(X)} P_1 \mbox{ where } P_1\hosub{A}{X}\SCongru P'$.
Because $P \SHOIOB Q$, $Q$ can simulate by
\annotate{\[
\begin{array}{l}
Q \st{\overline{a}(B)} \cdot \st{a(X)}  Q_1 \,\SHOIOB\, P_1 \\
\mbox{ where } ~ A \,\SHOIOB\, B \\
\mbox{ and } ~ Q_1\hosub{B}{X}\DEF  Q'
\end{array}
\]
}
$Q \st{\overline{a}(B)} \cdot \st{a(X)}  Q_1 \,\SHOIOB\, P_1, 
\mbox{ where } ~ A \,\SHOIOB\, B, 
\mbox{ and } ~ Q_1\hosub{B}{X}\DEF  Q'.
$
Now since the higher-order output is non-blocking, 
the two consecutive actions can contribute to forming a $\tau$ action, i.e., 
\annotate{\[
\begin{array}{l}
Q\st{\tau} Q'
\end{array}
\] 
}
$Q\st{\tau} Q'$, 
and we are left with showing $P' \SHOIOB Q'$. From $P_1 \,\SHOIOB\, Q_1$, Lemma \ref{l:process-sub-preserving}, we know 
$
P' \SCongru P_1\hosub{A}{X}  \,\SHOIOB\, Q_1\hosub{A}{X}\DEF  Q''.
$ 
By the congruence properties, we can derive due to $A \,\SHOIOB\, B$ that
$
Q'' \SCongru Q_1\hosub{A}{X} \,\SHOIOB\, Q_1\hosub{B}{X} \SCongru Q'.
$ Hence we conclude $P' \SHOIOB Q'$. 
%

\noindent\textbf{(2)} $P\st{\tau} P'$ comes from a component of $P$ alone. That is, $P\SCongru P_1\para P_2$, $P_1\st{\tau} P_1'$, and $P'\SCongru P_1'\para P_2$. Then we conclude by induction hypothesis. 
\end{proof}
\vspace*{-.6cm}


\subsection{Relating the strong bisimilarities}
We represent detailed relationship between the strong bisimilarities defined so far.
Such relationship will be established step by step.
Eventually, as our ultimate goal, it will be demonstrated that all these strong bisimilarities coincide with each other. 
This coincidence immediately entails that every and each of them is decidable, and moreover paves way for further discussion on the axiomatization and algorithm.
We first establish the coincidence between $\SHOB$ and $\SHOIOB$, and then move on to the remainder parts.

\vspace*{.1cm}
\noindent\textbf{\SHOB\ and \SHOIOB\ coincide}~~ 
%
\annotate{
We recall that $\hosub{\ve{R}}{\ve{X}}$ denotes pairwise substitution $\hosub{R_1}{X_1}\cdots \hosub{R_n}{X_n}$ for tuples $\ve{R}$ and $\ve{X}$ that are $R_1,...,R_n$ and $X_1,...,X_n$ respectively. 
As usual $\SHOB$ can be extended to open terms $P$ and $Q$ with $\fpv{P,Q}=\ve{X}$ in the following fashion.
\[
P \,\SHOB\, Q \quad\mbox{ if and only if }\quad P\hosub{\ve{R}}{\ve{X}} \,\SHOB\, Q\hosub{\ve{R}}{\ve{X}} \qquad\quad \mbox{ for any closed } \ve{R}
\]
}
\tdup{\scriptsize The extension to abstractions is similar, and actually this is already captured in the definition of strong HO bisimulation. \rc{So this is a built-in element of the strong HO bisimulations (and other related strong bisimulations as well). }}
The following lemma gives a characteristic of the strong HO bisimilarity. Its proof is a standard bisimulation deduction, 
\iftoggle{appendixON}{%
\rc{with details in Appendix \ref{a:proofs_1}.}
} {%
\xxyrmcolor{with details in \cite{XZ21L}.}
}
\begin{lemma}\label{l:exd2open4HOB}
Suppose $\fpv{P,Q} {=} \ve{X} {=} {X_1, ..., X_n}$.
For fresh $h$ and any closed $\ve{R}$, it holds that 
$
P\hosub{\ve{R}}{\ve{X}} \,\SHOB\, Q\hosub{\ve{R}}{\ve{X}} 
~\mbox{ if and only if }~ 
h(X_1).\cdots.h(X_n).P \,\SHOB\, h(X_1).\cdots.h(X_n).Q.
$
\end{lemma}

As Lemma \ref{l:coinHOandHOIO} states, the strong HO bisimilarity and the strong HO-IO bisimilarity are actually coincident. With the help of Lemma \ref{l:exd2open4HOB}, and Lemmas \ref{l:name-sub-preserving}, \ref{l:process-sub-preserving}, \ref{l:tau-preserving} as well, we can prove the mutual inclusion of the two strong bisimilarities.  
 \iftoggle{appendixON}{%
\rc{The details are provided in Appendix \ref{a:proofs_1}.}
} {%
\xxyrmcolor{The details are provided in \cite{XZ21L}.}
}
\begin{lemma}\label{l:coinHOandHOIO}
On \HOmp\ terms, $\SHOB \;=\; \SHOIOB$.
\end{lemma}



\noindent\textbf{Relating \SHOB\ and the other strong bisimilarities}~~ 
We now tackle the relationship between the strong HO bisimilarity and other strong bisimilarities, including the strong context bisimilarity ($\SCTXB$), strong normal bisimilarity ($\SNRB $), and open strong normal bisimilarity ($\OSNRB$), as well as the remaining part of the overall picture of coincidence. 
To do this, we need some preparation. 
The following lemma will be useful. 
The proof employs the usual bisimulation-establishing method. 
\iftoggle{appendixON}{%
\rc{The details are in Appendix \ref{a:proofs_1}.}
} {%
\xxyrmcolor{The details can be found in \cite{XZ21L}.}
}

\iftoggle{mNotsON}{%
} {%
\TODOM{expand...discussion...relationships of all the strong bisimilarities }\\
\TODOM{notice extension to abstraction and open processes...}
}



\iftoggle{mNotsON}{%
} {%
\TODOM{*****\bc{TO UPDATE from VerbTex} ...***** ... DN!!}
}




\iftoggle{mNotsON}{%
} {%
\TODOM{expand...discussion...relationships of all the strong bisimilarities ... DN }\\
\TODOM{notice extension to abstraction and open processes...DN}
}

\begin{lemma}\label{l:prop1_OSNRB}



Assume that $m$ is fresh with respect to $P_1$, $Q_1$, $P$ and $Q$, said otherwise $m\notin \fn{P_1,Q_1,P,Q}$. 
\noindent~\textbf{(1)} If $m.P_1\para P \,\OSNRB\, m.Q_1\para Q$, 
 then $P_1 \,\OSNRB\, Q_1$ and $P \,\OSNRB\, Q$.
\noindent~\textbf{(2)} If $m(Z).P_1\lrangle{Z} \para P \,\OSNRB\, m(Z).Q_1\lrangle{Z}\para Q$, then $P_1 \,\OSNRB\, Q_1$ and $P \,\OSNRB\, Q$.
\noindent~\textbf{(3)} If $m(Z).Z\lrangle{P_1} \para P \,\OSNRB\, m(Z).Z\lrangle{Q_1}\para Q$, then  $P_1 \,\OSNRB\, Q_1$ and $P \,\OSNRB\, Q$.
\end{lemma}

Using a similar proof strategy as Lemma \ref{l:prop1_OSNRB}, one can prove similarly the result for $\SNRB$.
\annotate{(we prove it for $\OSNRB$ because it is harder); hence the following corollary.} 
\begin{corollary} \label{cor:prop1_OSNRB2SNRB}
The result of Lemma \ref{l:prop1_OSNRB} also holds if one replaces $\OSNRB$ with $\SNRB$ in the statement. 
\end{corollary} 

\iftoggle{mNotsON}{%
} {%
\TODOM{cont...}

\TODOM{notice Ref...}

\TODOM{expand...more...}
}


Next in Lemma \ref{l:prop1_REL}, 
we present the first two implications about the strong bisimilarities. 
Basically, the proof of Lemma \ref{l:prop1_REL}\rc{(1)} utilizes the congruence of $\SHOB$ and the proof of Lemma \ref{l:prop1_REL}\rc{(2)} uses the fact that the requirements of $\SNRB$ are actually special cases of $\SCTXB$. 
\iftoggle{appendixON}{%
\rc{The details can be found in Appendix \ref{a:proofs_1}.}
} {%
\xxyrmcolor{The details can be found in \cite{XZ21L}.}
}
\begin{lemma}\label{l:prop1_REL}
\noindent\textbf{(1)} $\SHOB$ implies $\SCTXB$ on \HOmp\ processes. ~~\textbf{(2)} $\SCTXB$ implies $\SNRB$ on \HOmp\ processes. 
\end{lemma}
\vspace*{-.1cm}

\annotate{
\begin{lemma}\label{l:prop2_REL}
$\SCTXB$ implies $\SNRB$ on \HOmp\ processes. 
\end{lemma}
\vspace*{-.1cm}
}

We now demonstrate, in Lemma \ref{l:prop3_REL}, the two last inclusions about the strong bisimilarities, so as to finalize the jigsaw. The proofs of them use the usual bisimulation construction approach, by exploiting Lemmas \ref{l:lts_prop1},\ref{l:prop1_OSNRB} and Corollary \ref{cor:prop1_OSNRB2SNRB}.
\iftoggle{appendixON}{%
\rc{The details are placed in Appendix \ref{a:proofs_1}.}
} {%
\xxyrmcolor{The details are referred to \cite{XZ21L}.}
}
\begin{lemma}\label{l:prop3_REL}
\noindent\textbf{(1)} $\SNRB$ implies $\OSNRB$ on \HOmp\ processes. ~~\textbf{(2)} $\OSNRB$ implies $\SHOIOB$ on \HOmp\ processes.
\end{lemma}

\annotate{
\begin{lemma}\label{l:prop4_REL}
$\OSNRB$ implies $\SHOIOB$ on \BHOParam\ processes. 
\end{lemma}
}

\noindent The follow-up lemma essentially fills what is left in the relationship between the strong bisimilarities.
\begin{lemma}\label{l:prop5_REL}
$\SHOB$, $\OSNRB$, and $\SCTXB$ coincide on open and closed  \HOmp\ processes. 
\end{lemma}\vspace*{-.3cm}
\begin{proof}[Proof of Lemma \ref{l:prop5_REL}]

\iftoggle{mNotsON}{%
} {%
\TODOM{ TO REFER ... and ADAPT / REWRITE ...DN}
}

\annotate{ 
The following circular implications prove this lemma. 
\[
\begin{array}{rl}
& \SHOIOB \\
\,\st{\mbox{\scriptsize Lemma \ref{l:coinHOandHOIO}}}\, & \SHOB \\
\,\st{\mbox{\scriptsize Lemma \ref{l:prop1_REL}}}\, & \SCTXB \\
\,\st{\mbox{\scriptsize Lemma \ref{l:prop2_REL}}}\, & \SNRB \\
\,\st{\mbox{\scriptsize Lemma \ref{l:prop3_REL}(1)}}\, & \OSNRB \\
\,\st{\mbox{\scriptsize Lemma \ref{l:prop3_REL}(2)}}\, & \SHOIOB \\
\end{array}
\]
It can be shown in the following diagram.
 
\[ 
\xymatrix@C=30pt{
 \SHOIOB 
 \ar@{.>}[rr]^{\mbox{\scriptsize Lemma \ref{l:coinHOandHOIO}}} 
 \ar@/_2.8pc/@{<.}[0,8]|{\mbox{\scriptsize Lemma \ref{l:prop4_REL}}} 
 &&  \SHOB \ar@{.>}[rr]^{\mbox{\scriptsize Lemma \ref{l:prop1_REL}}} 
 && \SCTXB \ar@{.>}[rr]^{\mbox{\scriptsize Lemma \ref{l:prop2_REL}}}
 &&  \SNRB \ar@{.>}[rr]^{\mbox{\scriptsize Lemma \ref{l:prop3_REL}}} 
 && \OSNRB   \\
}
\]
}

The following circular implications prove this lemma.
\[ 
\xymatrix@C=30pt{
 \SHOIOB 
 \ar@{.>}[rr]^{\mbox{\scriptsize Lemma \ref{l:coinHOandHOIO}}} 
 \ar@/_2.8pc/@{<.}[0,8]|{\mbox{\scriptsize Lemma \ref{l:prop3_REL}\rc{(2)}}} 
 &&  \SHOB \ar@{.>}[rr]^{\mbox{\scriptsize Lemma \ref{l:prop1_REL}\rc{(1)}}} 
 && \SCTXB \ar@{.>}[rr]^{\mbox{\scriptsize Lemma \ref{l:prop1_REL}\rc{(2)}}}
 &&  \SNRB \ar@{.>}[rr]^{\mbox{\scriptsize Lemma \ref{l:prop3_REL}\rc{(1)}}} 
 && \OSNRB   \\
}
\]
\end{proof}
\vspace*{-.3cm}

\noindent\textbf{The main theorem}~~ 
From Lemma \ref{l:prop5_REL} and Lemma \ref{l:decidability_hoio_bisi}, we now have the main result of this section.
\begin{theorem}\label{l:prop6f_REL}
All the strong bisimilarities, that is, $\SHOIOB$, $\SHOB$, $\SCTXB$, $\SNRB$, and $\OSNRB$, coincide on open and closed \HOmp\ processes, and are all decidable. 
\end{theorem}
\annotate{
\begin{proof}
\iftoggle{mNotsON}{%
} {%
\TODOM{ TO REFER ... and ADAPT / REWRITE ...DN}
}

Directly from Lemma \ref{l:prop5_REL}. 
\end{proof}

The decidability result for the strong bisimilarities follows from Theorem \ref{l:prop6f_REL} and Lemma \ref{l:decidability_hoio_bisi}.
\begin{corollary}
All the strong bisimilarities, i.e., $\SHOIOB$, $\SHOB$, $\SCTXB$, $\SNRB$, and $\OSNRB$, are decidable.
\end{corollary}

}


\vspace*{-.6cm}
\section{Axiomatization}\label{s:axiomatization}
\vspace*{-.2cm}

In this section, we make an axiom system for the strong bisimilarities based on the decidability result. For simplicity, we denote by $\SB$ the strong bisimilarity, since all the strong bisimilarities coincide. Basically, the equation set of the axiom system is composed of the extended structural congruence. Compared with the setting without parameterization \cite{LPSS10a}, we have extra equations describing the application operation. We use a similar approach to prove the correctness of the axiom system.

%
{
\annotate{A complete axiomatization for original model without parameterization is given in \cite{LPSS10a}.}

The axiom system of \cite{LPSS10a} consists of the basic structural congruence laws and an extended distribution law
$DIS$: $a(x).\left(P\para \prod_{1}^{k-1}a(x).P\right) = \prod_{1}^{k}a(x).P$.
\annotate{
\begin{tabular}{cl}
\\
$(DIS)$ &\ $a(x).\left(P\para \prod_{1}^{k-1}a(x).P\right) = \prod_{1}^{k}a(x).P$
\\\\
\end{tabular}
}
Recall that the rules for application are modeled as a part of the structural congruence.  \xxyrmcolor{To admit parameterization, we introduce the following two more laws $APP1$ and $APP2$: 
}
$(\langle X\rangle P)\langle Q\rangle = P\{Q/X\}$, 
$(\langle x\rangle P)\langle m\rangle = P\{m/x\}$.
\annotate{
\begin{tabular}{cl}
\\
$(APP1)$ &\ $(\langle X\rangle P)\langle Q\rangle = P\{Q/X\}$
  \\
$(APP2)$ &\ $(\langle x\rangle P)\langle m\rangle = P\{m/x\}$
\\\\
\end{tabular}
}
We will show that the basic laws for the structural congruence, together with law $DIS$ and moreover laws $APP1$ and $APP2$, are sufficient for a complete axiom system. More specifically, we first prove that any term $P$ has a unique \emph{prime} decomposition $\prod^{k}_{i{=}1} P_i$, and then that any term can be simplified to a normal form (denoted as $\normalform{P}$) \xxyrmcolor{by the laws above}. 
\xxyrmcolor{By the soundness of the laws, we infer $P\sim \normalform{P}$.} Finally we prove that for any $P$ and $Q$, $P \sim Q$ if and only if $\normalform{P} \SCongru \normalform{Q}$.
}

\eat{\scriptsize 
Outline of several key points (to see if reusable...):
\begin{enumerate}
\item Unique decomposition:  Any term $P$ has a \emph{prime} decomposition $\prod^{k}_{i{=}1} P_i$.
\item The distribution law:
\[a(X).(P\para \prod^{k{-}1}_{i{=}1} a(X)P) = \prod^{k}_{i{=}1} a(X)P
\]
\item Normal form: a term form that cannot be further simplified using the extended structural congruence (i.e., the standard structural congruence extended with the distribution law (from left to right))
\item For any $P$, $P$ is strongly bisimilar to its normal form. That is, $P \sim \normalform{P}$.
\item For any $P$ and $Q$ in normal form, $P \sim Q$ iff $P\SCongru Q$.
\item For any $P$ and $Q$, $P \sim Q$ iff $\normalform{P}\SCongru \normalform{Q}$.
\item ...
\item ...
\end{enumerate}
}


{Let $\mathcal{A}$ be the axiom system containing $DIS$, $APP1$, $APP2$ and the commutative monoid laws for parallel composition. 
We will prove the completeness of $\mathcal{A}$ in the remainder of this section. We start by the cancellation property.} The point of proving this property is to deem $\SB$ as $\SHOIOB$ \cite{LPSS10a}, and the most involved cases are those concerning the abstractions. 
\iftoggle{appendixON}{%
\rc{We provide the details in Appendix \ref{a:proofs_2}.}
} {%
\xxyrmcolor{We provide the details in \cite{XZ21L}.}
}
\begin{proposition}[Cancellation]\label{pps:cancellation}
For all $P,Q$ and $R$, if $P\para R \SB Q \para R$ then $P\SB Q$.
\end{proposition}


The notion of prime processes is due to \cite{MM93,LPSS10a}. 
A process $P$ is \emph{prime} if $P\not\SB 0$ and $P\SB P_1 \para P_2$ implies $P_1 \SB 0$ or $P_2 \SB 0$.
If $P\SB \prod_{i=1}^{n}P_i$ where each $P_i$ is prime, we call $\prod_{i=1}^{n}P_i$ is a \emph{prime decomposition} of $P$.
The following proposition states that for any process, there is a unique prime decomposition up to the strong bisimilarity and permutation of indices.
Instantiating $\SB$ as $\SHOIOB$, one can prove this proposition by induction on the size of the given process. 
\iftoggle{appendixON}{%
\rc{We give the proof in Appendix \ref{a:proofs_2}.}
} {%
\xxyrmcolor{We give the proof in \cite{XZ21L}.}
}
\begin{proposition}[Unique prime decomposition]\label{prop:UniDecomp}
Given a process $P$, if there are two prime decompositions $P\SB\prod_{i=1}^{n}P_i$ and $P\SB\prod_{j=1}^{m}Q_j$, then $n=m$ and there is a permutation $\sigma:\{1,2,\ldots,n\}\rightarrow\{1,2,\ldots,n\}$, such that $P_i \sim Q_{\sigma(i)}$ for each $i \in \{1,2,\ldots,n\}$.
\end{proposition}

\eat{
\begin{proof}

We proceed by induction on $\pdepth{P}$.
\begin{itemize}
	\item If some $P_i \SB Q_j$ (w.l.o.g., assume that $P_1 \SB Q_1$), then we have the following two prime decompositions for $P$: $P\SB P_1 \para \prod_{i=2}^{k}P_i$ and $P\SB Q_1 \para \prod_{j=2}^{l}Q_j$. By Proposition \ref{pps:cancellation}, we have $\prod_{i=2}^{k}P_i \SB \prod_{j=2}^{l}Q_j$. By induction hypothesis, the two prime decompositions $\prod_{i=2}^{k}P_i $ and $\prod_{j=2}^{l}Q_j$ are identical up to $\SB$ and permutation of indices. Thus $\prod_{i=1}^{k}P_i $ and $\prod_{j=1}^{l}Q_j$ are also identical.
	\item Assume that for every $i,j$ we have $P_i \nsim Q_j$.
	\begin{itemize}
		\item If either $k=1$ or $l=1$, then $k=l=1$ and $P_1=Q_1$ by the definition of prime process. This is a contradiction.
		\item If $k,l \geq 2$, w.l.o.g., we can assume that $\pdepth{P_1} \leq \pdepth{P_i}$ for any $1\leq i \leq k$ and $\pdepth{P_1} \leq \pdepth{Q_j}$ for any $1\leq j \leq l$. 
		\begin{itemize}
			\item If $P_1 = X$, as $P \SB \prod_{j=1}^{l}Q_j$, then one of $Q_j$ must be $X$, a contradiction.
			\item \rc{If $P_1 = X\lrangle{A}$, as $P \SB \prod_{j=1}^{l}Q_j$, then one of $Q_j$ must be $X\lrangle{B}$ with $B \SB A$, we thus have $X\lrangle{A}\SB X\lrangle{B}$, a contradiction.}
			\item \rc{If $P_1 = X\lrangle{d}$, as $P \SB \prod_{j=1}^{l}Q_j$, then one of $Q_j$ must be $X\lrangle{d}$, a contradiction.}
			\item If $P_1=m(X).R$. Since $\pdepth{R} < \pdepth{P}$, by induction hypothesis, $R$ has a unique prime decomposition $R = \prod_{g=1}^{h}R_g$. We have 
			$P= \prod_{i=2}^{k}P_i  \para m(x).(\prod_{g=1}^{h}R_g) \st{m(X)}P'$
			with unique prime decomposition $P'\sim\prod_{g=1}^{h}R_g \para \prod_{i=2}^{k}P_i$. Since $P\SB\prod_{j=1}^{l}Q_j$, w.l.o.g., the corresponding transition is $\prod_{j=1}^{l}Q_j\st{m(X)}T \para \prod_{j=2}^{l}Q_j \sim P'$. By induction hypothesis, the prime decomposition of $T \para \prod_{j=2}^{l}Q_j$ should be $\prod_{g=1}^{h}R_g \para \prod_{i=2}^{k}P_i$. As $Q_2$ is prime, it must be equal with a process in $\bigcup_{g=1}^{h}R_g \cup \bigcup_{i=2}^{k}P_i$. Since $\pdepth{R_g}<\pdepth{Q_2}$ for any $1\leq g \leq h$, by Lemma \ref{l:depth_sim}, $R_g \nsim Q_2$ for any $1\leq g \leq h$, thus $Q_2 \sim P_i$ for some $2\leq i \leq k$, a contradiction with the assumption that $P_i \nsim Q_j$ for every $i,j$.
			\item If $P_1=\overline{m}(R)$. Similar to the last case. 
		\end{itemize}
	\end{itemize}
\end{itemize}
\end{proof}
\sepp
}

We write $P \rewriteDIS Q$ if there are $P'$ and $Q'$ such that $P\SE P'$, $Q\SE Q'$, and $Q'$ can be obtained from $P'$ by rewriting a {subterm} of $P'$ by laws $DIS$, {$APP1$, $APP2$} from left to right.
A process $P$ is in \emph{normal form} if it cannot be simplified by using $\rewriteDIS$. Any process $P$ has a unique normal form up to $\SE$, denoted as $\normalform{P}$.
It is not hard to derive the following property.
\begin{lemma}\label{lm:DISsoundness}
If $P \rewriteDIS Q$, then $P \SB Q$. For any $P$, $P \SB \normalform{P}$.
\end{lemma}
\annotate{
\begin{proof}[Proof of Lemma \ref{lm:DISsoundness}]
Define relation
 $$\mathcal{R} = \rewriteDIS \cup \rewriteDIS^{-1} \cup \SE$$
The proof is done by showing that $\mathcal{R}$ is a bisimulation.
\end{proof}
}

\annotate{Following \cite{LPSS10a}, we prove a similar lemma, which is crucial for the completeness proof.}
Next we give a lemma crucial for the completeness proof. Its counterpart in non-parameterization setting was first presented in \cite{LPSS10a}. 
\iftoggle{appendixON}{%
\rc{The proof of Lemma \ref{lm:primeNmform} is put in Appendix \ref{a:proofs_2}.}
} {%
\xxyrmcolor{The proof of Lemma \ref{lm:primeNmform} can be found in \cite{XZ21L}.}
}
\begin{lemma}\label{lm:primeNmform}
If $a(X).P{\SB} Q \para Q'$ ($Q,Q'\nsim 0$), then $a(X).P{\SB} \prod_{i=1}^{k}a(X).A$ ($k{>}1$) with $a(X).A$ in normal form.
\end{lemma}
\vspace*{-.1cm}

\zwbrmcolor{Now we can prove the completeness of the axiom system $\mathcal{A}$.}
\xxyrmcolor{
	 Basically, the proof of the completeness uses a similar approach as that of \cite{LPSS10a}. The main novelty here is to accommodate the parameterization in the equation system 
	 and the corresponding parts in the induction, i.e., those parts concerning the abstraction and application for names and processes.
}
\begin{lemma}[Completeness]\label{lem:complete}
For any $P,Q$, if $P\SB Q$ then $\normalform{P}\SCongru \normalform{Q}$.
\end{lemma}\vspace*{-.5cm}
\begin{proof}[Proof of Lemma \ref{lem:complete}]
We first show the following two properties simultaneously:
\noindent\textbf{1.} If $A$ is a prefixed process in normal form, then $A$ is prime;
\noindent\textbf{2.} For any \rc{$A,B$} in normal form, $A \SB B$ implies $A \SE B$.
We proceed by induction on $\pdepth{A}$. The case $\pdepth{A}= 0$ is immediate as the only term of this size is $\Nil$. Suppose the property holds for all $\pdepth{A}< n$ with $n\geq 1$.

	\noindent\textbf{(1)} Assume $A$ is of the form $a(X).A'$. Suppose $A$ is not prime, $A \SB P_1\para P_2$. By Lemma \ref{lm:primeNmform}, $A \SB \prod_{i=1}^{k}a(X).B$ with $k>1$ and $a(X).B$ in normal form. Then $A' \SB B \para \prod_{i=1}^{k-1}a(X).B$. By ind. hyp. for 2, we have $A' \SE B \para \prod_{i=1}^{k-1}a(X).B$. Then $A \SE a(X).(B \para \prod_{i=1}^{k-1}a(X).B)$, a contradiction to that $A$ is in normal form.

	\noindent\textbf{(2)} Suppose $A \SB B$, we proceed by a case analysis on the structure of $A$.

		\noindent\textbf{$\bullet$} $A$ is $X$. We have that $B$ should be the same variable $X$.
		
		\noindent\textbf{$\bullet$} $A$ is $m(X).P$. Assume $B$ is not prime, $B \SB P_1\para P_2$. By Lemma \ref{lm:primeNmform}, we know $A \SB \prod_{i=1}^{k}a(X).Q$ with $k>1$ and $a(X).Q$ in normal form. But according to property 1, $A$ is prime, a contradiction. We thus have $B$ is $m(X).Q$ with $P \SB Q$. By ind. hyp., $P\SE Q$. We thus have $A\SE B$.
		
		\noindent\textbf{$\bullet$} $A$ is $\overline{m}(Q)$. We have that $B$ is $\overline{m}(Q')$ with $Q\sim Q'$. By ind. hyp., $Q\SE Q'$.  We thus have $A\SE B$.
		
		\noindent\textbf{$\bullet$} $A$ is $\lrangle{X}P$. We have that $B$ is $\lrangle{X}Q$ and $P\SB Q$. By ind. hyp., $P\SE Q$.  We thus have $A\SE B$.
		
		\noindent\textbf{$\bullet$} $A$ is $\zwbrmcolor{X}\lrangle{Q}$. We have that $B$ is $X\lrangle{Q'}$ and $Q\SB Q'$. By ind. hyp., $Q\SE Q'$.  We thus have $A\SE B$.
		
		\noindent\textbf{$\bullet$} $A$ is $\lrangle{x}P$. We have that $B$ is $\lrangle{x}Q$ and $P\SB Q$. By ind. hyp., $P\SE Q$.  We thus have $A\SE B$.
		
		\noindent\textbf{$\bullet$} $A$ is $\zwbrmcolor{X}\lrangle{n}$.  We have that $B$ is $X\lrangle{n}$, and then $A\SE B$.
		
		\noindent\textbf{$\bullet$} $A$ is $\prod_{i=1}^{k}P_i$ with $k>1$ and $P_i$ is not a parallel composition. We discuss over the possible shape of $P_i$. 

			\noindent~~\textbf{-} If there exists $j$ s.t. $P_j = X$, then $B \SE X \para B'$. Thus $A \SE B$ follows by ind. hyp. on $\prod_{1\leq i \leq k, i\neq j}P_i$ and $B'$. 

			\noindent~~\textbf{-} If there exists $j$ s.t. $P_j = X\lrangle{Q}$, then $B \SE X\lrangle{Q'} \para B'$ with $Q \SB Q'$ and $B'\SB \prod_{1\leq i \leq k, i\neq j}P_i$. By ind. hyp., $A\SE B$.
			
			\noindent~~\textbf{-} If there exists $j$ s.t. $P_j = X\lrangle{n}$, then $B\SE X\lrangle{n} \para B'$ with $B' \SE \prod_{1\leq i \leq k, i\neq j}P_i$. By ind. hyp., $A\SE B$.
			
			\noindent~~\textbf{-} If there exists $j$ s.t. $P_j = \overline{m}(Q)$, then $B$ must contain an output component on the same channel. We thus have $B = \overline{m}(Q')\para B'$ with $Q'\SB Q$ and $B'\SB \prod_{1\leq i \leq k, i\neq j}P_i$. By ind. hyp., $Q'\SE Q$ and $B'\SE \prod_{1\leq i \leq k, i\neq j}P_i$, which implies $A \SE B$.
			
			\noindent~~\textbf{-} The last case is $A = \prod_{i=1}^{k}m_i(X_i).P_i $. According to property 1, each component $ m_i(X_i).P_i$ is prime. Similarly, $B \SE \prod_{i=1}^{l}n_i(Y_i).Q_i $ and each component $ n_i(Y_i).Q_i$ is prime. By Proposition \ref{prop:UniDecomp}, $k=l$ and $ m_i(X_i).P_i\SB n_i(Y_i).Q_i$ for $1\leq i\leq k$ (up to a permutation of indices). By ind. hyp. $P_i \SE Q_i$ for all $i$, which finally implies $A \SE B$.


Now for $P,Q$, assume $P\SB Q$. Let $A \DEF \normalform{P}$ and $B \DEF \normalform{Q}$. By Lemma \ref{lm:DISsoundness}, $A\SB P\SB Q \SB B$. As $A, B$ are in normal form, have $A\SE B$, and then $ \normalform{P} \SE  \normalform{Q}$, as needed.
\end{proof}



\vspace*{-.6cm}
\section{Algorithm for the bisimilarity checking}\label{s:dec_algorithm}
\vspace*{-.2cm}
In this section, based on the results in the previous sections, we develop an algorithm for checking the strong bisimilarity. 
We utilize the tree approach proposed in \cite{LPSS10a}, i.e., encoding a \HOmp\ process as a tree, normalizing this tree to be compared up-to syntax. Differently now, the tree and the normalization takes parameterization into consideration. 
We define a function $\db$ that assigns De Bruijn indices to variables\cite{Bru72,LPSS10a}.
Here the variables include the ones introduced by input prefixed processes, name abstraction and process abstraction. 
Following \cite{LPSS10a}, we introduce the representation of a term by a tree. We write $t[m_1,\ldots,m_k]$ for a tree with label $t$ and subtrees $m_1,\ldots,m_k$.

\begin{definition}[Tree representation]
The tree representation of $P$ is defined inductively as follows.

	\noindent\textbf{(1)} $\tree(0) = 0[\ ]$,
	\noindent~\textbf{(2)} $\tree(X) = \db(X)[\ ]$,
	\noindent~\textbf{(3)} $\tree(a(X).P) = \labelI{a}[\tree(P)]$,\\
	\noindent~\textbf{(4)} $\tree(\overline{a}(Q)) = \labelO{a}[\tree(Q)]$,
	\noindent~\textbf{(5)} $\tree(x(X).P) = \labelI{\db(x)}[\tree(P)]$,
	\noindent~\textbf{(6)} $\tree(\overline{x}(Q)) = \labelO{\db(x)}[\tree(Q)]$,\\
	\noindent~\textbf{(7)} $\tree(\prod_{i=1}^n P_i) = \textsf{par}[\tree(P_1),\ldots,\tree(P_n)]$,
	\noindent~\textbf{(8)} $\tree(\lrangle{X}P) = \tnAbs[\tree(P)]$,
	\noindent~\textbf{(9)} $\tree(\lrangle{P}Q) = \tnApp[\tree(P), \tree(Q)]$,
	\noindent~\textbf{(10)} $\tree(\lrangle{x}P) = \tnAbs[\tree(P)]$,
	\noindent~\textbf{(11)} $\tree(\lrangle{P}n) = \tnApp[\tree(P), n[\ ]]$.
\end{definition}

\vspace*{-1mm}
The algorithm deciding the strong bisimilarity depends on the following 3 normalization steps: \\
\noindent\textbf{Normalization}:  
\noindent\textbf{(1)}~ 
In the first step, \xxyrmcolor{the term is rewritten by the application} rules $APP1$, $APP2$ if possible.
\noindent\textbf{(2)}~ 
The second step focuses on the normalization of parallel composition. W.l.o.g., we can assume that the children of parallel composition nodes are not parallel composition nodes. After this step, every parallel composition node has at least two sorted child nodes, and none of them is $0$. 
\noindent\textbf{(3)}~
The last step aims to apply $DIS$ from left to right if possible.


Now \xxyrmcolor{we explain the detailed algorithms given as pseudocodes below}.  
A tree node $n$ has the following attributes: $n.type$ for the type of corresponding process, the values can be $zero$, $var$, $inp$, $out$, $par$, $abs$, $app$; $n.label$ for the label of the tree node; $n.numChildren$ for the number of children nodes; $n.children$ for the lists of all child nodes.
The algorithm App realizes the application operation. It requires three parameters: $n_{raw}$, $ind$, and $n_{eval}$. The tree is traversed top-down and all variables from term $n_{raw}$ are replaced with process $n_{eval}$ if the De Bruijn index matches $ind$. 
In the process of application, if there are more than one occurrence of an abstracted variable, say $X$, to be replaced, there will be more than one duplications of $n_{eval}$.  The nests of application may result in an exponential explosion on the number of tree nodes. However, we can make optimization by reusing $n_{eval}$, that is, each occurrence of $X$ points to the same tree of $n_{eval}$. Then it is guaranteed that the space cost for normalized terms is still linear, leading to acceptable time complexity of the algorithm.

Algorithm NS1 deals with terms for application. The tree is traversed bottom-up. Every term in the form of $\langle X\rangle P)\langle Q\rangle$ or $\langle x\rangle P)\langle m\rangle$ are rewritten as $P\{Q/X\}$ or $P\{m/x\}$ respectively. Terms in the form of $X\lrangle{Q}$ and $X\lrangle{n}$ remain unchanged.
Algorithm NS2 deals with parallel composition. First all zero processes are removed. Then, if attribute $numChildren$ is $0$, the tree is collapsed to a zero node. If attribute $numChildren$ is $1$, the tree is collapsed to its single child. After this, all  children nodes are sorted.
In algorithm NS3, the tree is traversed bottom-up to find subtrees which can apply $DIS$ from left to right. Lines $\ref{NS3:line:PMstart}$-$\ref{NS3:line:PMend}$ decides if node $n$ matches the pattern with the left-hand side of $DIS$. It harnesses the property that all children nodes have been sorted in normalization step 2. If it fails to match the pattern, the node $n$ remains unchanged and the function returns in line $\ref{NS3:line:return1}$ or $\ref{NS3:line:return2}$. 
Otherwise, the term is rewritten at lines $\ref{NS3:line:RWstart}$-$\ref{NS3:line:RWend}$.
As a consequence of Lemma \ref{lem:complete}, the following lemma shows that, if two terms are strongly bisimilar, they can be normalized to the same tree by the three normalization steps. By checking the equalities of the two trees, we can decide the strong bisimilarity between \HOmp\ terms.
\vspace*{-.1cm}
\begin{lemma}
Let $P, Q$ be two terms. Let $T_P$, $T_Q$ be the tree representations of $P,Q$ respectively. Assume that $T'_P$, $T'_Q$ are the normalized trees after the normalization steps $1$-$3$. Then $P\sim Q$ if and only if $T'_P= T'_Q$.
\end{lemma}
\annotate{
\begin{proof}
Immediate from Lemma \ref{lem:complete}.
\end{proof}
}

{
\renewcommand{\thealgorithm}{}
\floatname{algorithm}{Application}
\begin{algorithm*}[b]
\footnotesize
		\caption{\footnotesize App($n_{raw}$,$ind$,$n_{eval}$)}
		\label{step1}
		\begin{multicols}{2}
		\begin{algorithmic}[1]
			\Require 
			Tree nodes $ n_{raw} $, $n_{eval}$, an integer $ind$.
			\State \textbf{if}  ($ n_{raw} $.type == `var' \textbf{or} $n_{raw} $.type == `inp') \textbf{and} $ n_{raw} $.label == $\labelI{ind}$  \textbf{then}
			\State \quad $n_{raw} =n_{eval}$
			\State \textbf{end if}
			\State \textbf{if}  $ n_{raw} $.type == `out' \textbf{and} $ n_{raw} $.label == $\labelO{ind}$  \textbf{then}
			\State \quad $n_{raw} =n_{eval}$
			\State \quad $n_{raw}$.label = ($n_{raw}$.label)$^O$
			\State \textbf{end if}
			\State \textbf{if}  $ n_{raw} $.type == `inp' \textbf{or} $ n_{raw} $.type == `abs' \textbf{then}
			\State \quad $ind = ind +1$
			\State \textbf{end if}
			\State \textbf{for} $ i $ = 1 \textbf{to} $ n $.numChildren \textbf{do}
			\State \quad App($ n_{raw} $.children[$i$], $ind$ , $n_{eval}$)
			\State \textbf{end for}
	    \end{algorithmic}
		\end{multicols}\vspace*{-.4cm}
	\end{algorithm*}

\setcounter{algorithm}{0}
\renewcommand{\thealgorithm}{\arabic{algorithm}}
\floatname{algorithm}{Normalization Step}
\begin{algorithm*}[t]
\footnotesize
		\caption{\footnotesize NS1($ n $)}
		\label{step1}
		\begin{multicols}{2}
		\begin{algorithmic}[1]
			\Require
			A tree node $ n $
			\State \textbf{for} $ i $ = 1 \textbf{to} $ n $.numChildren \textbf{do}
			\State \quad NS1($ n $.children[$i$])
			\State \textbf{end for}
			\State \textbf{if}  $ n $.type == `app' \textbf{then}
			\State \quad \textbf{if} $ n $.children[1].type == `abs' \textbf{then}
			\State \quad \quad App($ n $.children[1].children[1], 1, $ n $.children[2])
			\State \quad \textbf{end if}
			\State \textbf{end if}
	    \end{algorithmic}
		\end{multicols}\vspace*{-.35cm}
	\end{algorithm*}

	\begin{algorithm*}[!h]
	\footnotesize
		\caption{\footnotesize NS2($ n $)}
		\label{step2}
		\begin{multicols}{2}
		\begin{algorithmic}[1]
			\Require
			A tree node $ n $
			\State \textbf{for} $ i $ = 1 \textbf{to} $ n $.numChildren \textbf{do}
			\State \quad NS2($ n $.children[$i$])
			\State \textbf{end for}
			\State \textbf{if}  $ n $.type == `par' \textbf{then}
			\State \quad $ j $ = 1
			\State \quad \textbf{for} $ i $ = 1 \textbf{to} $ n $.numChildren \textbf{do}
			\State \quad \quad \textbf{if} $ n $.children[$ i $].type $ \neq $ `zero' \textbf{then}
			\State \quad \quad \quad $ n $.children[$ j $] = $ n $.children[$ i $]
			\State \quad \quad \quad $ j $ = $ j $ + 1
			\State \quad \quad \textbf{end if}
			\State \quad \textbf{end for}
			\State \quad $ n $.numChildren = $ j $ - 1		
			\State \quad \textbf{if} $ n $.numChildren == 0 \textbf{then}
			\State \quad \quad $ n $.type = `zero'
			\State \quad \textbf{end if}
			\State \quad \textbf{if} $ n $.numChildren == 1 \textbf{then}
			\State \quad \quad $ n $ = $ n $.children[1]
			\State \quad \textbf{end if}
			\State \textbf{end if}
			\State sortChildren($ n $)
	    \end{algorithmic}
	    \end{multicols}\vspace*{-.4cm}
	\end{algorithm*}

\begin{algorithm*}[!h]
\footnotesize
	\caption{\footnotesize NS3($ n $)}
	\label{step3}
	\begin{multicols}{2}
	\begin{algorithmic}[1]
		\Require
		A tree node $ n $
		\State \textbf{for} $ i $ = 1 \textbf{to} $ n $.numChildren \textbf{do}
		\State \quad NS3($ n $.children[$ i $])
		\State \textbf{end for}
		\State \textbf{if}  $ n $.type == `inp' \textbf{then}
		\State \quad $ p $ = $ n $.children[1]
		\State \quad \textbf{if}  $ p $.type == `par' \textbf{then}
		\State \quad \quad $ smallIndex$=-1
		\State \quad \quad $small$,$big$=\textbf{null}
		\State \quad \quad $pc1$=$p$.children[1]
		\State \quad \quad $pc2$=$p$.children[$p$.numChildren]
		\State \quad \quad \textbf{if} $pc1$.type == `inp' \textbf{and} $pc1$.label == $ n $.label \textbf{and} $pc1$.children[1] == $pc2$ \textbf{then} \label{NS3:line:PMstart}
		\State \quad \quad \quad $ small $ = $pc2$, $ big $ = $pc1$, 
		\State \quad \quad \quad $ smallIndex $ = $ p $.numChildren
		\State \quad \quad \textbf{else if} $pc2$.type==`inp' \textbf{and} $pc2$.label == $ n $.label \textbf{and} $pc2$.children[1]==$pc1$ \textbf{then}
		\State \quad \quad \quad $ small $ = $pc1$, $ big $ = $pc2$
		\State \quad \quad \quad $ smallIndex $ = 1
		\State \quad \quad \textbf{else} 
		\State \quad \quad \quad \textbf{return}  \label{NS3:line:return1} 
		\State \quad \quad \textbf{end if}
		\State \quad \quad \textbf{for} $ i $ = 2 \textbf{to} $ n $.numChildren-1 \textbf{do}
		\State \quad \quad \quad \textbf{if} $ p $.children[$ i $] $ \neq big$  \textbf{then}
		\State \quad \quad \quad \quad \textbf{return} \label{NS3:line:return2} 
		\State \quad \quad \quad \textbf{end if}
		\State \quad \quad \textbf{end for} \label{NS3:line:PMend}
		\State  \quad \quad $ p $.children[$ smallIndex $] = $ big $ \label{NS3:line:RWstart} 
		\State  \quad \quad $ n=n $.children[$1$] \label{NS3:line:RWend} 
		\State \quad \textbf{end if}
		\State \textbf{end if}
	\end{algorithmic}
	\end{multicols}\vspace*{-.4cm}
\end{algorithm*}
}

\eat{
\begin{algorithm*}[htb]
	\caption{NS3($ n $)}
	\label{step3}
	\begin{algorithmic}[1]
		\Require
		A tree node $ n $
		\State \textbf{for} $ i $ = 1 \textbf{to} $ n $.numChildren \textbf{do}
		\State \quad NS3($ n $.children[$ i $])
		\State \textbf{end for}
		\State \textbf{if}  $ n $.type == `inp' \textbf{then}
		\State \quad $ p $ = $ n $.children[1]
		\State \quad \textbf{if}  $ p $.type == `par' \textbf{then}
		\State \quad \quad $ smallIndex $ = -1, $small$=\textbf{null}, $big$=\textbf{null}, $pc1$=$p$.children[1], $pc2$=$p$.children[2]
		\State \quad \quad \textbf{if} $ pc1 $ == $ pc2 $ \textbf{and}
		$ pc1 $.type ==`inp' \textbf{and}
		$pc1$.label == $ n $.label \textbf{then}
		\State \quad \quad \quad $ small $ = $pc1$.children[1], $ big $ = $pc1$
		\State \quad \quad \textbf{else if} $pc1$.type == `inp' \textbf{and} $pc1$.label == $ n $.label \textbf{and} $pc1$.children[1] == $pc2$ \textbf{then}
		\State \quad \quad \quad $ small $ = $pc2$, $ big $ = $pc1$, $ smallIndex $ = 2
		\State \quad \quad \textbf{else if} $pc2$.type == `inp' \textbf{and} $pc2$.label == $ n $.label \textbf{and} $pc2$.children[1] == $pc1$ \textbf{then}
		\State \quad \quad \quad $ small $ = $pc1$, $ big $ = $pc2$, $ smallIndex $ = 1
		\State \quad \quad \textbf{else} 
		\State \quad \quad \quad \textbf{return} 
		\State \quad \quad \textbf{end if}
		\State \quad \quad \textbf{for} $ i $ = 3 \textbf{to} $ n $.numChildren \textbf{do}
		\State \quad \quad \quad \textbf{if} $ p $.children[$ i $] $ \neq big$  \textbf{and} $ p $.children[$ i $] $ \neq small$ \textbf{then}
		\State \quad \quad \quad \quad \textbf{return}
		\State \quad \quad \quad \textbf{else if} $ p $.children[$ i $] == $small$  \textbf{and} $ smallIndex  \neq -1 $  \textbf{then}
		\State \quad \quad \quad \quad \textbf{return}
		\State \quad \quad \quad \textbf{else if} $ p $.children[$ i $] == $small$  \textbf{and} $ smallIndex$  == $-1$  \textbf{then}
		\State \quad \quad \quad \quad $ smallIndex  = i $ 
		\State \quad \quad \quad \textbf{else}
		\State \quad \quad \quad \quad \textbf{continue}
		\State \quad \quad \quad \textbf{end if}
		\State \quad \quad \textbf{end for}
		\State \quad \quad \textbf{if} $ smallIndex \neq -1$ \textbf{then}
		\State \quad \quad \quad $ p $.children[$ smallIndex $] = $ big $
		\State \quad \quad \quad $ n=n $.children[$1$]
		\State \quad \quad \textbf{end if}
		\State \quad \textbf{end if}
		\State \textbf{end if}
	\end{algorithmic}
\end{algorithm*}
}


We now analyze the complexity of the algorithm. Given processes $P$ and $Q$, let $n$ be the sum of the number of nodes in the tree representations of $P$ and $Q$.
The algorithm $App$ and $NS1$ traverse the tree for one time and can be done in $O(n)$ time. The most time-consuming part of $NS2$ is sorting, which can be done in $O(n\log(n))$ time. The algorithm $NS3$ can be performed in $O(n)$ time. Therefore, bisimilarity checking takes in $O(n\log(n))$ time in total.
As explained above, the space complexity is $O(n)$. 



%
\vspace*{-.5cm}
\section{Conclusion}\label{s:conclusion}
\vspace*{-.3cm}

In this paper, we have exhibited that even in presence of parameterization, which can increase the expressiveness of higher-order processes, the strong bisimilarity is still decidable for \HOmp. The proving approach extends the previous one for HOcore, with several significant distinctions due to parameterization. This decidability result comes with the more powerful modelling capability of the process model, and is thus of both fundamental and practical importance to some extent. Besides, an axiom system and an algorithm are provided. They can be used as an intermediate prototype for potential application of the higher-order process model, in particular the bisimilarity checking. 
\annotate{\bc{For simplicity, we have focused on unary abstraction, i.e., abstraction with only one formal variable. The results can be extended to the general case, by some tedious but manageable adjustments. }
}
A further work is to try expanding the model to allow more convenient modelling capability, e.g., locations, while maintaining the decidability result. A far more challenging job is to consider the decidability of the weak bisimilarity. \annotate{, either to obtain decidability or prove undecidability, and in the former to design and analyse the bisimilarity checking algorithm.}



\annotate{
\subsection*{Future work} 
There are some issues worthy of further examination. 
\begin{itemize}
\item A further work is to try expanding the model to allow more convenient modelling capability, e.g., locations, while maintaining the decidability result. A far more challenging job is to consider the weak bisimilarity, either to obtain decidability or prove undecidability, and in the former to design and analyse the bisimilarity checking algorithm.

\item Weak bisimilarity?

\item 
Also one may consider extending to allowing infinite application. So that $\Omega$ can be defined ($\Omega \DEF O\lrangle{O}$ in which $O\DEF \lrangle{X}(X\lrangle{X})$), and replication can be derived ($!P\DEF O'\lrangle{O'}$ in which $O'\DEF \lrangle{X}(P\para X\lrangle{X})$). See if the results of this paper carry over, including decidability. It seems that the at least some part of the technical approach needs adjustment, if possible; (e.g., process-substitution-preserving lemma). Or if undecidability indeed emerges, give an explicit proof, and then hopefully look for a decidable fragment. 

\end{itemize}	
}



\sepp
\noindent\textbf{Acknowledgements}\;\;
We are grateful to the anonymous referees for their useful comments on this paper. 
\vspace*{-.4cm}
\bibliographystyle{eptcs}
\bibliography{process}

\iftoggle{appendixON}{%
\sepp\sepp
\clearpage
\appendix
\noindent\textbf{\Large Appendix}


\input{appendix_proof_1.tex}
\input{appendix_proof_2.tex}
}{%
}


\end{document}